\documentclass[11pt,letterpaper]{article}

\usepackage[utf8]{inputenc}
\usepackage[english]{babel}

\usepackage[dvipsnames]{xcolor}
\usepackage[colorlinks=true,pdfpagemode=UseNone,urlcolor=RoyalBlue,linkcolor=RoyalBlue,citecolor=OliveGreen,pdfstartview=FitH]{hyperref}
\definecolor{hkured}{HTML}{EE4123}
\definecolor{hkublue}{HTML}{009BD4}
\definecolor{hkugreen}{HTML}{00B38C}
\definecolor{hkuyellow}{HTML}{FED401}

\usepackage{framed}
\definecolor{shadecolor}{HTML}{E0E0E0}

\usepackage{tikz}
\usetikzlibrary{shapes, arrows.meta, positioning, patterns}
\usepackage{pgfplots}
\usepgfplotslibrary{fillbetween}
\pgfdeclarelayer{ft}
\pgfdeclarelayer{bg}
\pgfsetlayers{bg,main,ft}
\pgfplotsset{compat=1.15}

\usepackage{subcaption}
\usepackage{tabto}

\usepackage[margin=1in]{geometry}
\usepackage{amsfonts,amsmath,amsthm,amssymb}
\usepackage{mathtools,thmtools}
\usepackage{bm}

\usepackage{array}
\newcolumntype{x}[1]{>{\centering\arraybackslash\hspace{0pt}}p{#1}}

\usepackage{nicefrac}

\declaretheorem[name=Informal Theorem]{informal}
\declaretheorem[parent=section]{theorem}
\declaretheorem[sibling=theorem]{lemma}
\declaretheorem[sibling=theorem]{corollary}

\declaretheorem[style=definition,parent=section]{definition}
\declaretheorem[style=definition,parent=section]{example}

\usepackage{todonotes}
\usepackage{tcolorbox}
\usepackage{booktabs}
\usepackage{caption}
    
\usepackage{multirow}

\usepackage{enumitem}

\usepackage[numbers, sort]{natbib}
\newtcolorbox{algorithm}[1]
{
	adjusted title = {#1},
	fonttitle = \bfseries,
  	beforeafter skip = 12pt,
}
\usepackage{cleveref}

\renewcommand{\Pr}{\operatorname{\mathbf{Pr}}}
\newcommand{\E}{\operatorname{\mathbf{E}}}
\newcommand{\Var}{\operatorname{\mathbf{Var}}}
\newcommand{\defeq}{\stackrel{\mathrm{def}}{=}}

\newcommand{\dif}[1]{\mathrm{d}#1}

\newcommand{\Z}{\mathbb{Z}}
\newcommand{\R}{\mathbb{R}}

\newcommand{\load}{L}
\newcommand{\loadset}{\mathcal{L}}
\newcommand{\loadeset}{\mathcal{S}}
\newcommand{\polytope}{\mathcal{P}}
\newcommand{\support}{\mathrm{span}}

\newcommand{\Poisson}{\mathrm{Poisson}}
\newcommand{\Bernoulli}{\mathrm{Bern}}
\newcommand{\Exponential}{\mathrm{Exp}}
\newcommand{\Swap}{\mathrm{Swap}}
\newcommand{\MarkSwap}{\mathrm{MarkSwap}}

\newcommand{\MarkTime}{T}

\title{Generalized Balls into Bins}
\author{
   Zhiyi Huang \footnote{The University of Hong Kong. Email: zhiyi@cs.hku.hk.}
   \and
   Kaifeng Lin \footnote{SKLP, Institute of Computing Technology, CAS. Email: linkaifeng25s@ict.ac.cn. Some work was done when the author was at Tsinghua University.}
   \and
   Qinpei Lou \footnote{The University of Hong Kong. Email: qinpeilou@connect.hku.hk. Some work was done when the author was at Zhejiang University.}
   \and
   Xinyue Xiang \footnote{University of Macau. Email: yc67195@um.edu.mo. Some work was done when the author was at Peking University.}
   \and
   Peilin Yang \footnote{The University of Hong Kong. Email: peilinyang@connect.hku.hk.}
}
\date{August 2026}

\begin{document}

\begin{titlepage}
\thispagestyle{empty}

\maketitle

\begin{abstract}
\thispagestyle{empty}
Consider a set of bins and two-choice balls arriving by a Poisson process. 
We must allocate each incoming ball immediately to one of two incident bins. 
For a given function $f$ and every bin, we aim to bound the expectation of $f(\load)$---where $\load$ is the bin's final load---based on the arrival rate of balls incident to that bin.
We call this problem Generalized Balls into Bins, capturing many problems as special cases including the original Balls into Bins by Azar et al.~(1994) and Online Stochastic Matching by Feldman et al.~(2009). 
We show that Greedy provides optimal amortized bounds for all convex and concave functions $f$.
Further, we propose another algorithm that achieves non-trivial bounds without amortization.
As an application, we design a competitive algorithm for a stochastic model of completion time minimization on unrelated machines.
\end{abstract}	

\end{titlepage}

\section{Introduction}
\label{sec:intro}

Thirty years ago, \citet*{AzarBKU:SODA:1994} introduced the celebrated Balls into Bins problem about online allocation of $n$ balls into $n$ bins.
If we allocate each ball independently to a random bin, a simple calculation shows that the most loaded bin will have about $\log n/\log\log n$ balls with high probability.
By contrast, \citet{AzarBKU:SODA:1994} showed that the maximum load decreases to about $\log \log n$ if we randomly sample two bins independently for each ball and then allocate the ball to the less loaded bin.
This surprising result is widely known as ``the power of two choices''.

The two-choice paradigm has found numerous applications in hashing~\cite{KarpLM:STOC:1992,Wieder:FTTCS:2017,CelisRSW:SICOMP:2013}, load balancing~\cite{AdlerCMR:STOC:1995,Mitzenmacher:FOCS:1996}, routing~\cite{ColeMMMRSSV:STOC:1998}, fair division~\cite{EdmondsP:FOCS:2006}, etc.
We refer readers to \citet*{MitzenmacherRS:2001} for the early results on this topic, and \citet{Wieder:FTTCS:2017} for a more recent survey.

Fifteen years later, \citet*{FeldmanMMM:FOCS:2009} proposed the Online Stochastic Matching problem---a stochastic variant of the Online Bipartite Matching problem by \citet*{KarpVV:STOC:1990}.
The problem has applications in the online allocation of impressions to advertisers based on distributional information, e.g., learned from historical data.
Their algorithm samples two suggested matchings between the advertisers and the possible impressions based on the distributional information.
Then, for every arrived impression, it first attempts to allocate the impression according to the first matching;
if that advertiser is already matched, it then tries the second matching.
With this approach, \citet{FeldmanMMM:FOCS:2009} obtained a competitive ratio greater than the $1-\frac{1}{e}$ barrier, and credited their idea to the ``power of two choices'' framework.

Despite the similarity between the two-choice Balls into Bins algorithm by \citet{AzarBKU:SODA:1994} and the two-choice Online Stochastic Matching algorithm by \citet{FeldmanMMM:FOCS:2009}, no overarching framework has yet been developed to formally capture these two problems as special cases.

\subsection{Conceptual Contribution}

This paper introduces the Generalized Balls into Bins problem as such a framework, inspired by the graphical model of Balls into Bins~\cite{KenthapadiP:SODA:2006,PeresTW:RSA:2015,BansalF:STOC:2022} and the recent progress on Online Stochastic Matching~\cite{ChenHS:FOCS:2024,Yan:SODA:2024,QiuFZW:WINE:2023}.
We start by explaining these works.

The graphical model of Balls into Bins considers a set of bins represented by the vertices of an undirected graph $G = (V, E)$, and a set of two-choice balls represented by the edges.
The model proceeds in $T$ steps.
At each step, nature samples a ball/edge uniformly at random;
then, the algorithm allocates it to an incident bin, aiming to minimize the maximum load of the bins at the end.
\citet{KenthapadiP:SODA:2006}, \citet{PeresTW:RSA:2015}, and \citet{BansalF:STOC:2022} generalized the classical Balls into Bins results to the graphical model, under various structural conditions on the graph.
It is helpful to consider an alternative Poisson model, where two-choice balls $\{u, v\}$ arrive with rate $\lambda_{\{u,v\}}$.
By standard arguments (see \Cref{sec:asymptotic-equivalence}), the original graphical model is asymptotically equivalent to the Poisson model with rate $\lambda_{\{u,v\}} = \frac{T}{|E|}$ if $\{u, v\} \in E$, and $0$ otherwise.
%

On the other hand, several recent algorithms for Online Stochastic Matching adopted a relax-and-round approach.
For each incoming impression/ball, the algorithm sub-samples two potential advertisers/bins based on the solution to a linear program (LP) relaxation, independent of previous arrivals and allocations. 
This can be seen as a half-integral allocation or a two-choice ball incident to two bins.
Since the balls arrive by a Poisson process, the sub-sampled two-choice balls are also Poisson.
Then, a rounding algorithm allocates the ball to one of the two sub-sampled bins.
\citet{ChenHS:FOCS:2024} designed such rounding algorithms, which they called Stochastic Online Correlated Selection, leading to improved algorithms for Online Stochastic Matching and several generalizations.
The rounding for matching can be seen as a Balls into Bins algorithm but with a different objective: maximizing the matching size, i.e., the sum of the indicators for the bins' non-emptiness.

As a comprehensive framework, the Generalized Balls into Bins problem also considers a set of bins and a set of two-choice balls arriving by a Poisson process.
When each ball arrives, the online algorithm must immediately allocate it to an incident bin.
As a baseline, the Random algorithm allocates each ball to a random incident bin.
If a bin $v$'s (normalized) degree is $d_v$, defined as half the sum of its incident balls' arrival rates, then its load $\load_v$ follows the Poisson distribution with rate $d_v$, denoted as $\Poisson(d_v)$.
We are interested in algorithms that ``balance the loads''.

How should we measure balancedness to capture both the maximum load in Balls into Bins, and the bins' non-emptiness in Online Stochastic Matching?
Our first approach considers the specific function $f$ that we want to maximize/minimize, and ensures that the expectation of $f(\load_v)$ is at least/at most a bound $g(d_v)$ that depends on the bin's degree $d_v$.
The bound should be strictly better than the baseline, i.e., the expectation of $f(\load_v)$ when $\load_v$ follows $\Poisson(d_v)$.
In many problems, it is sufficient to consider the sum of $f(\load_v)$ over all bins $v$. 
If an algorithm ensures $\sum_v f(\load_v)$ is at least/at most $\sum_v g(d_v)$, we say that it guarantees an \emph{amortized} bound $g$.

Instead of focusing on a specific function $f$, we can also measure balancedness through the notion of second-order stochastic dominance, which is widely adopted in Economics for measuring risks.
For some distributions $D(d_v)$ that are parameterized by the degree $d_v$ and dominate the baseline $\Poisson(d_v)$, we ensure that the distribution of bin $v$'s load dominates $D(d_v)$.
This is a stronger notion than the first one, in the sense that it implies bound $g(d_v) = \E_{\load_v \sim D(d_v)} f(\load_v)$ for all concave/convex functions $f$ simultaneously.


\begin{shaded}
\vspace{-6pt}
\begin{example}[Matching]
	\label{exa:matching}
    The maximization problem with concave function $f(\load) = \min \big\{ \load, 1 \big\}$ corresponds to maximizing the cardinality of a matching between balls and bins, with each bin matched to the first ball allocated to it.
    We will further study a generalization that considers function $f(\load) = \min \big\{ \load, c \big\}$ for some positive integer capacity $c$.
    This corresponds to a $c$-matching problem where each bin can be matched up to $c$ times.
\end{example}

\begin{example}[Completion Time Minimization]
	\label{exa:completion-time}
    Suppose that each ball is a unit-size job and the incident bins are machines that can process the job.
    The algorithm allocates jobs arriving online to incident machines, after which the machines process the allocated jobs one by one.
    Minimizing the total completion time of the jobs corresponds to the minimization problem with convex function $f(\load) = \frac{\load^2 + \load}{2}$, since the completion time of the $\load$ jobs allocated to a machine are $1$ to $\load$.
\end{example}


\begin{example}[Load Balancing]
    Again, let the balls be unit-size jobs, and the bins be machines.
    The Balls into Bins problem considers minimizing the maximum load $\max_{v \in V} \load_v$.
    It is known that the softmax function:
    \[
        \frac{1}{\eta} \ln \sum_{v \in V} e^{\eta \load_v}
    \]
    is a $\frac{1}{\eta} \ln |V|$ additive approximation to the maximum load.
    Indeed, bounding the expectation of $\sum_{v \in V} e^{\eta \load_v}$ is a crucial step in the analysis by \citet{TalwarW:ICALP:2014}.
    Further, if one could bound the expectation of $\sum_{v \in V} e^{\eta \load_v}$, a bound on the expected maximum load would follow since:
    \[
    	\E \frac{1}{\eta} \ln \sum_{v \in V} e^{\eta \load_v} ~\le~ \frac{1}{\eta} \ln \E \sum_{v \in V} e^{\eta \load_v}~.
    \]
%
	\vspace{-20pt}
\end{example}
\end{shaded}

The literature has also considered the heavily loaded regime with a large $\bar{\load} = \E \load_v$ for all bins, where the goal is to bound the maximum deviation from the expectation $\max_{v \in V} \big| \load_v - \bar{\load} \big|$.
	This can be captured by considering the following potentials in the literature:
	\[
		\sum_{v \in V} e^{\eta |\load_v - \bar{\load}|} ~\le~ \sum_{v \in V} \Big( e^{\eta ( \load_v - \bar{\load} )} + e^{\eta ( \bar{\load} - \load_v )} \Big)
		~.
	\]
	We remark that the Poisson arrival model is unsuitable for making progress in the heavily loaded regime, because the maximum deviation is strictly worse than the known bounds in the discrete-time model, due to the difference between the realized number of balls and its expectation.

\subsection{Technical Contribution}

\paragraph{Amortized Bounds via Greedy.}
We first study the simple Greedy algorithm that allocates each ball to the less loaded incident bin.
It is easy to find examples for which Greedy fails to give a non-trivial bound without amortization.
Hence, we will focus on amortized bounds.

\begin{informal}[\Cref{cor:concave-amortized,cor:convex-amortized}]
	For bounded degree instances, Greedy ensures for any convex/concave function $f$ that $\sum_v \E f(\load_v)$ is at least as good as:
	\[
		\sum_v \, \E_{j \sim \mathrm{Poisson}(2d)} \frac{f( \lfloor \nicefrac{j}{2} \rfloor ) + f( \lceil \nicefrac{j}{2} \rceil )}{2d}
        \cdot d_v
        ~.
	\] 
\end{informal}

The above bound matches the best achievable in the simple instance of two bins and one ball between them with rate $2d$. 
In this sense and to our surprise, Greedy is optimal simultaneously for all convex and concave functions on degree-$d$ instances.
We stress that instances with bounded degrees, especially those with $d = 1$, are of prominent importance in all three examples above.


\paragraph{Potential-based Amortized Analysis.}
We prove the above amortized bounds by a potential-based analysis, which may be of independent interest.
We construct a potential $\Phi_e(t)$ for each ball $e$, and consider an overall potential function $\Phi(t) = \sum_e \Phi_e(t)$.
We show that the sum of the algorithm's objective and the potential function is monotone over time, and the potential function diminishes to zero by the end of the time horizon.
This means the algorithm's expected objective at the end is at least the initial potential function value $\Phi(0)$.

The potential of a ball $e$ at time $t$ depends on the loads of $e$'s incident bins.
It captures how much the incident bins' function values $f(\load_v)$ would increase, should ball $e$ arrive in the remaining time horizon with rate $2d$, the maximum rate possible given the bounded degree, and assuming no arrivals of the other balls.
See \Cref{sec:greedy} for the precise definition.

This is in part inspired by the recent potential-based and differential-inequality-based analyses by \citet{HuangWINE:2024,HuangSY:STOC:2022}, in the sense that their analyses and ours both involve proving a quadratic inequality about the arrival rates of balls incident to a bin.
Compared to the counterparts in prior works, our potential is simpler and has an intuitive meaning as described above.
Our quadratic inequality is also simpler and reduces to Chebyshev's sum inequality. 
The simplicity and cleanness allow us to handle general convex/concave functions, while prior works only apply to matching.

\paragraph{Approximate MDP Algorithm.}
Generalized Balls into Bins can be seen as a Markov decision process (MDP), and thus, the optimal online policy is characterized by a dynamic program with an exponentially large state space.
From this perspective, we can interpret $\Phi(t)$ as an approximation to the value-to-go/cost-to-go function of the dynamic program, satisfying the Hamilton–Jacobi–Bellman equation with inequality.
As a result, an approximate MDP algorithm that replaces the value-to-go/cost-to-go function with $\Phi(t)$ gives the same optimal amortized bounds.

Further, this draws an interesting connection to a long line of research on approximate dynamic programming algorithms for MDP from Operations Research~\cite{Adelman:OPRE:2007,deFariasVR:OPRE:2003,ZhangA:TS:2009,deFariasVR:2007,Zhang:MSOM:2011,MeissnerS:EJOR:2012,TongT:JoC:2014,KunnumkalT:MOR:2016,ZhangSZ:OPRE:2022,MaRST:OPRE:2020,Jiang:2023,LiuVR:MSOM:2008}.
Most of them studied the network revenue maximization problems, which can be viewed as stochastic models of online packing, related to but different from our problem.
Nevertheless, rephrasing them in our context, they approximate the value-to-go function by a linear combination of basis functions, and almost all of them consider one basis function for each resource/bin $v$, i.e., $\Phi(t) = \sum_v \Phi_v(t)$.
Our work, by contrast, constructs a basis function for each ball/pair of bins, and demonstrates the effectiveness of these bases with the optimal amortized bounds for a broad class of Generalized Balls into Bins problems.
Intuitively, the ball-wise basis functions capture the pairwise substitutability of bins, i.e., ``the power of two choices,'' while the bin-wise basis functions in previous works cannot.
We believe this novel class of basis functions may lead to further related research in the literature on approximate MDP algorithms.

\paragraph{Non-amortized Algorithm.}
In \Cref{sec:mark-swap}, we complement the amortized optimality of Greedy by designing an algorithm called Mark-and-Swap that achieves non-trivial bounds for general convex/concave functions without amortization, through second-order stochastic dominance.

\begin{informal}[\Cref{lem:mark-swap-worst-case} and \Cref{thm:mark-swap-worst-case}]
	For some distributions $\MarkSwap(d)$ that stochastically dominate $\Poisson(d)$ for any $d \ge 0$, Mark-and-Swap ensures that the load $\load_v$ of any bin $v$ stochastically dominates $\MarkSwap(d_v)$.
\end{informal}

The table below compares the bounds given by the baseline Random algorithm and by the Mark-and-Swap and Greedy algorithms, when the bins have degrees $1$.
We examine the three functions from the aforementioned examples. 
For load balancing, we let $\eta = 1$ for ease of demonstration.

\begin{table}[h]
\renewcommand{\arraystretch}{1.3}
\centering
\begin{tabular}{llx{2.7cm}x{2.7cm}x{2.7cm}}
	\toprule
	& & Random & Mark-and-Swap & Greedy {\footnotesize (amortized)} \\
	\midrule
	Matching & $f(\load) = \min \{\load, 1\}$ & $1-\frac{1}{e}\approx 0.63$ & $1-\frac{2}{e^2} \approx 0.73$ & $1-\frac{2}{e^2} \approx 0.73$\\
	Completion time & $f(\load) = \frac{\load^2+\load}{2}$ & $1.5$ & $\frac{3}{2} - \frac{1}{e ^ {2}}\approx 1.36$ & $\frac{21}{16} - \frac{1}{16e^4} \approx 1.31$ \\
	Load balancing & $f(\load) = e^{\load}$ & $e^{e-1} \approx 5.57$ & $\approx 4.76$ & $\approx 3.89$ \\
	\bottomrule
\end{tabular}	
\end{table}

Furthermore, the non-amortized analysis of Mark-and-Swap is easier to generalize than the amortized analysis of Greedy.
Indeed, we apply it to obtain a $1.435$-competitive algorithm for the Online Stochastic Completion Time Minimization problem (\Cref{thm:completion-time}), where jobs of different sizes arrive online and are allocated to a set of unrelated machines.
By contrast, independently rounding Skutella's quadratic program~\cite{Skutella:JACM:2001} is only $1.5$-competitive.
Besides replacing independent rounding with Mark-and-Swap, we also introduce an LP relaxation for the stochastic model, borrowing ideas from the Online Stochastic Matching problem.
See \Cref{sec:completion-time} for details.

\subsection{Future Directions}


\paragraph{Multi-choice Balls.}
This paper focuses on two-choice balls due to their historic significance in the Balls into Bins and Online Stochastic Matching problems.
Nonetheless, it would be interesting to study allocating multi-choice balls to bins.
\Cref{sec:discussion} presents some preliminary results, including a generalization of Greedy and its amortized optimality to $k$-choice balls in $c$-matching, and a discussion on balls with mixed numbers of choices.

The superiority of multi-choice algorithms is known for matching.
\citet{HuangSY:STOC:2022} showed that $2$-choice algorithms are at best $0.706$-competitive, and gave a state-of-the-art $0.716$-competitive multi-choice algorithm. 
We leave it for future research the design of algorithms that allocate multi-choice balls effectively, and the decomposition of LP solution into multi-choice balls (see \Cref{sec:completion-time} and \citet{ChenHS:FOCS:2024} for some examples of decompositions into two-choice balls).

\paragraph{Improved Bounds for Well Connected Graphs.}
This paper parameterizes the bounds by the degree of the bin, which is local information that does not capture the global structure of the graph.
By contrast, \citet{BansalF:STOC:2022} studied the graphical Balls into Bins problem and found that an algorithm's ability to balance the loads depends on the connectedness of the graph. 
Indeed, the results in this paper do not provide an asymptotic improvement over the $\log n / \log\log n$ bound for the maximum load, as such an improvement is impossible for an arbitrary graph.
It is a natural and important next step to consider algorithms with better bounds for well-connected graphs.
Progress along this line will likely also lead to better algorithms for problems such as matching and completion time minimization.

\paragraph{Non-homogeneous Poisson Arrivals.}
\citet{ChenHS:FOCS:2024} gave a two-choice Stochastic Online Correlated Selection algorithm, which in our context provides the same bound of $1 - \frac{2}{e^2}$ as Mark-and-Swap and Greedy for $f(\load) = \min \{\load, 1\}$ when the bins have degrees $1$.
Further, their bound holds more generally for non-homogeneous Poisson arrivals.
On one hand, the baseline Random algorithm and the resulting load distribution $\Poisson(d)$ extend to the non-homogeneous model.
On the other hand, the analyses of Greedy and Mark-and-Swap in this paper fail to generalize.
Designing algorithms with non-trivial bounds for non-homogeneous arrivals is another interesting research direction.

\subsection{Other Related Work}

\paragraph{Balls into Bins.}
\citet{BerenbrinkCSV:STOC:2000} and \citet{TalwarW:ICALP:2014} studied the heavily loaded regime, where there are much more balls than bins.
\citet{Mitzenmacher:TPDS:2001} and \citet{ColeFMMRSU:1998} considered models in which balls depart over time.
Last but not least, \citet{Vocking:JACM:2003} and \citet{Godfrey:SODA:2008} investigated algorithms for multi-choice balls.

\paragraph{Online Stochastic Matching.}
Since \citet{FeldmanMMM:FOCS:2009}, there has been a series of results~\cite{BahmaniK:ESA:2010,HaeuplerMZ:WINE:2011,ManshadiOS:MOR:2012,JailletL:MOR:2014,BrubachSSX:ESA:2016,HuangS:STOC:2021} on two-choice algorithms for Online Stochastic Matching.
The state-of-the-art competitive ratios for unweighted matching are $0.716$ for homogeneous arrivals via a multi-choice algorithm~\cite{HuangSY:STOC:2022} and $0.69$ for non-homogeneous arrivals via a two-choice algorithm~\cite{ChenHS:FOCS:2024}.
\citet{ChenHS:FOCS:2024} also gave algorithms with competitive ratios better than $1-\frac{1}{e}$ for Display Ads and AdWords in the non-homogeneous Poisson model.
We refer readers to \citet{Mehta:FTTCS:2013} and \citet{HuangTW:SIGecom:2024} for further references on online matching and related problems.

\paragraph{Completion Time Minimization.}
Online Stochastic Completion Time Minimization studies online scheduling for completion time minimization, with unrelated machines and stochastically generated jobs released one by one at time $0$.
The case of zero or common release time has been studied for various scheduling problems in both the offline setting~\cite{LeungLPZ:IPL:2007,Skutella:JACM:2001,SchulzS:SIDMA:2002} and the online setting~\cite{BampisKLP:ISAAC:2023,BampisLMZ:COCOON:2012,AlbersMS:SPAA:2007}.
There were no prior results on the stochastic model considered in this paper.

Interestingly, the $1.5$ ratio from independently rounding Skutella's quadratic program is also a natural barrier for offline completion time minimization.
This barrier was first broken by \citet{BansalSS:STOC:2016} using dependent rounding, followed by a series of improvements~\cite{Li:SICOMP:2020,ImS:SODA:2020,ImL:SODA:2023,Harris:SODA:2024,Li:SODA:2025}.

\section{Preliminaries}
\label{sec:prelim}

For any real number $x \in \R$, we write $\lfloor x \rfloor$, $\lceil x \rceil$, and $\{ x \}$ for the greatest integer less than or equal to $x$, the smallest integer greater than or equal to $x$, and the fractional part of $x$.
We write $z^+$ for function $\max \{ z, 0 \}$.


We write $X \sim D$ to denote that a random variable $X$ follows distribution $D$.
We will use $X$ and $D$ interchangeably for notational convenience when the meaning is clear from the context.
For example, $\E X = \E D$ and $\Var X = \Var D$ are the expectation and variance of random variable $X$/distribution $D$.
For any distribution $D$ and any real-valued function $f$ defined on the support of $D$, we write $f(D) \defeq \E_{X \sim D} f(X)$ for the expectation of the function.

Let $\Poisson(\lambda)$ and $\Exponential(\lambda)$ be the Poisson and exponential distributions with rate $\lambda > 0$.
Let $\Bernoulli(p)$ be the Bernoulli distribution that takes value $1$ with probability $0 \le p \le 1$.

For any two distributions $D_1$ and $D_2$, let $D_1 + D_2$ denote the distribution of $X_1 + X_2$ with $X_1$ drawn from $D_1$ and $X_2$ drawn from $D_2$ independently.
Similarly, for any distribution $D$ and any positive integer $k \ge 1$, let $k \cdot D$ denote the distribution of $X_1 + X_2 + \dots + X_k$ with $X_i$ drawn from $D$ independently for $1 \le i \le k$.

A random variable $X \sim D_X$ \emph{second-order stochastically dominates} another random variable $Y \sim D_Y$, written as $X \succeq Y$ or $D_X \succeq D_Y$, if $\E f(X) \ge \E f(Y)$ for any non-decreasing concave function $f$.
If the two variables have the same expectations, we have an equivalent definition:
there is a random variable $Z$ such that (1) $X+Z$ follows the same distribution as $Y$, and (2) $\E \,[\, Z \mid X \,] = 0$ for any realization of $X$.
Following the terminology in Statistics and Economics, we say that $Y$ is a \emph{mean-preserving spread} of $X$.
See, e.g., \citet{RothschildS:JET:1970}.

\section{Model of Generalized Balls into Bins}

Consider a set of bins $V$ and a set of two-choice balls $E = \binom{V}{2}$.
Balls $e \in E$ arrive by a Poisson process with rates $\lambda_e \ge 0$ in time horizon $[0, 1]$.
The notations $V$ and $E$---inherited from the graphical model of Balls into Bins~\cite{BansalF:STOC:2022}---suggest that we interpret the balls and bins as the edges and vertices of a weighted undirected graph, whose weights are the rates of the balls.

When a ball $e = \{u, v\}$ arrives, the algorithm must immediately allocate it to an incident bin, i.e., either $u$ or $v$.
We write $e \sim v$ or $v \sim e$ to denote that ball $e$ is incident to bin $v$.
We further define the (normalized) \emph{degree} of a bin $v \in V$ as:
\[
    d_v \defeq \frac{1}{2} \sum_{e \sim v} \lambda_e
    ~.
\]


Let $\load_v$ be a random variable denoting the number of balls allocated to a bin $v \in V$ at the end, which we will refer to as bin $v$'s \emph{load}.
We want to maximize (respectively, minimize) the expected sum of a concave (respectively, convex) function $f : \Z_{\ge 0} \to \R_{\ge 0}$ applied to the bins' loads, i.e.:
\[
    \E \, \sum_{v \in V} f(\load_v)
    ~.
\]

We only consider convex functions $f$ whose expectations $f\big(\Poisson(d)\big)$ are well defined for the baseline distribution $\Poisson(d)$ for any degree $d \ge 0$.
%
%
Consider any minimization problem with a convex function $f$. 
We say that an online algorithm guarantees a bound $g : \R_{\ge 0} \to \R_{\ge 0}$ for function $f$, if for every bin $v \in V$ we have:
\[
	\E f(\load_v) \le g(d_v)
	~.
\]

We say that it guarantees an amortized bound $g$ if:
\[
	\E \sum_{v \in V} f(\load_v) \le \sum_{v \in V} g(d_v)
	~.
\]


%
%

We define the same concepts for maximization problems with concave functions $f$ by changing the directions of the inequalities.

\subsection{Extension: Beyond Two Choices}

We can extend the model to study balls incident to $k$ bins.
Consider a set of bins $V$ and the set of $k$-choice balls $E = \binom{V}{k}$.
Balls $e \in E$ arrive by a Poisson process with rates $\lambda_e$ in time horizon $[0, 1]$.
In other words, we consider a weighted undirected $k$-uniform hypergraph, whose hyperedges and vertices correspond to the $k$-choice balls and bins.

The definitions in the basic model directly apply, except that the degree of a bin $v \in V$ is normalized differently:
\[
	d_v = \frac{1}{k} \sum_{e \sim v} \lambda_e
	~.
\]

Even more generally, we may consider a model in which balls might be incident to different numbers of bins.
In this case, we view the balls and bins as the hyperedges and vertices of a possibly non-uniform hypergraph.
We define the degree of a bin $v \in V$ to be:
\[
	d_v = \sum_{e \sim v} \frac{1}{|e|} \lambda_e
	~,
\]
where $|e|$ denotes the size of a hyperedge $e$.

\subsection{Extension: Combinatorial Functions}

The next extension associates each bin with a combinatorial function over multisets of incident balls.
We want to minimize or maximize the sum of the bins' function values over the allocated balls.
Consider a set of bins $V$ and a set of $2$-choice balls $E$.
Note that $E$ may not be $\binom{V}{2}$ in this extension because we allow multiple edges/balls between two vertices/bins if they are distinct for the combinatorial functions.

While we leave the general form of this extension for future research, we highlight two special combinatorial functions that are linear combinations of the functions from \Cref{exa:matching,exa:completion-time} over nested subsets of incident balls.

\begin{example}[Edge-weighted Matching with Free Disposal, a.k.a., Display Ads]
	Let there be a positive weight $w_{e,v} \ge 0$ between any bin $v$ and any incident ball $e \sim v$.
	A bin $v$'s combinatorial function maps a multiset $S$ of incident balls allocated to it to the maximum weight, i.e.:
	\[
		f_v(S) \,=\, \max_{e \in S} w_{e,v}
		~.
	\]
	This corresponds to a maximum weight bipartite matching problem, where each bin is matched to the heaviest ball allocated to it.
	
	We may rewrite this objective by considering bin $v$'s load for each weight-level $w > 0$:
	\[
		\load_v(w) \defeq \Big| \big\{ e \in S : w_{e,v} \ge w \big\} \Big|
		~,
	\]
	i.e., the number of balls allocated to $v$ whose weights are at least $w$.
	Then, we have:
	\[
		f_v(S) = \int_0^\infty \min \big\{ \load_v(w) , 1 \big\} \,\dif{w}
		~.
	\]
\end{example}

\begin{example}[Completion Time Minimization with General Job Sizes]
	Let $s_{e,v} \ge 0$ denote the processing time required for machine $v$ to complete job $e$.
	Given a set of jobs $S$ allocated to a machine $v$, processing the jobs in ascending order of their sizes, known as Smith's Rule, minimizes their total completion time.
	Suppose there are $k$ jobs in $S$ with sizes $s_1 \le s_2 \le \dots \le s_k$.
	The minimized total completion time is:
	\[
		f_v(S) = s_1 + (s_1+s_2) + \dots + (s_1 + s_2 + \dots + s_k)
		~.
	\]
	
	Let $\load_v(s)$ be the number of jobs of sizes at least $s$ that are allocated to $v$.
	The above objective can be rewritten as:
	\begin{equation}
		\label{eqn:completion-time-by-size}
		f_v(S) = \int_0^\infty \frac{\load_v(s)^2 + \load_v(s)}{2} \,\dif{s}
		~.
	\end{equation}
\end{example}

\section{Greedy Algorithm and Potential-based Amortized Analysis}
\label{sec:greedy}

This section considers the simple Greedy algorithm defined as follows.

\begin{algorithm}{Greedy Algorithm}
    Allocate each ball to the less loaded incident bin, breaking ties uniformly at random.
\end{algorithm}

The Greedy algorithm has been studied extensively in the literature.
For a 1-regular complete graph, the seminal ``power of two choices'' paper by \citet{AzarBKU:SODA:1994} showed that Greedy yields a maximum load at most $\log \log |V| + O(1)$ with high probability;
At the same time, the baseline Random algorithm gets maximum load at most $(1 + o(1))\ln |V| / \ln \ln |V|$ with high probability.
Further, the Greedy algorithm resembles the Balance algorithm for the online $c$-matching problem and its generalization (e.g., \citet{KalyanasundaramP:TCS:2000}).

\subsubsection*{Impossibility of Non-amortized Bounds}

We first remark that Greedy cannot give any non-trivial bound without amortization. 
Consider a star graph with a root $r$ and $|V|-1$ leaves. 
For any leaf $\ell$, let the arrival rate of edge $\{r, \ell\}$ be $2$.
Although the degrees of the leaves are $1$, Greedy would allocate almost two balls to each leaf in expectation when $V$ is sufficiently large, because when each ball arrives the incident leaf would be less loaded than the root with high probability.
Hence, from each leaf's perspective, it gets a higher load from Greedy than from the baseline Random algorithm.

That being said, the heavier loads of the leaves are compensated by the lighter load of the root in this example.
The overall objective of Greedy is better than the baseline by Random.
Therefore, this section will focus on proving amortized bounds.

\subsubsection*{Potential-based Amortized Analysis}

Let $A(t)$ be a random variable that denotes Greedy's objective at time $t \in [0, 1]$. 
The next lemma summarizes the conditions needed by a potential-based amortized analysis.

\begin{lemma}
    \label{lem:potential-function-conditions}
    Consider any maximization (resp., minimization) problem with function $f : \Z_{\ge 0} \to \R_{\ge 0}$ and any bound $g : \R_{\ge 0} \to \R_{\ge 0}$.
    Suppose
    $\Phi: [0, 1] \to \R_{\ge 0}$ is a potential function such that:
    \begin{enumerate}
        \item $\Phi(0) \ge \sum_{v \in V} g(d_v)$ (resp., $\Phi(0) \le \sum_{v \in V} g(d_v)$);
        \hspace*{\fill} (initial condition)
        \item $\E \big[ \Phi(t) + A(t) \big]$ is non-decreasing (resp., non-increasing) over $t \in [0, 1]$;
        \hspace*{\fill} (monotonicity)
		\item $\Phi(1) = 0$.          
		\hspace*{\fill} (terminal condition)
    \end{enumerate}
    Then, the algorithm guarantees an amortized bound $g$.
\end{lemma}
\begin{proof}
	The lemma for a maximization problem follows by:
    \begin{align*}
        \E A(1) 
        &
        = \E \bigl[ \Phi(1) + A(1) \bigr] 
        \tag{terminal condition} \\[.75ex]
        &
        \ge \E \bigl[ \Phi(0) + A(0) \bigr]
        \tag{monotonicity} \\[.75ex]
        &
        = \Phi(0) 
        \tag{$A(0) = 0$ by definition}\\
        &
        \ge \sum_{v \in V} g(d_v)
        ~.
        \tag{starting condition}
    \end{align*}
    
    The counterpart for minimization problems is almost identical, except having the inequalities in the opposite direction.
\end{proof}

\subsection{Analysis Framework}
\label{sec:2-choice-framework}


We next present a framework for analyzing the two‑choice setting in the basic model.
It naturally generalizes to both multi-choice and mixed‑choice settings, which we defer to \Cref{sec:general-framework}.

\paragraph{Separable Potential Function.}
We will design a potential for each ball $e$, denoted as $\Phi_e : [0, 1] \to \R_{\ge 0}$.
Then, we will consider the overall potential obtained by summing over all balls:
\[
    \Phi(t) \,\defeq\, \sum_{e \in E} \lambda_e \cdot \Phi_e(t)
	~.
\]

Recall that $\load_v$ denotes the load of bin $v$, i.e., the number of balls in it.
Denote the load pair of each ball $e = \{u, v\}$ by an unordered pair $S_e = (\load_u, \load_v)$.
The potential of ball $e$ at time $t$ depends on its current load pair. 
That is, for functions $\Phi^{S} : [0, 1] \to \R_{\ge 0}$ to be determined, we let:
\[
	\Phi_e(t) = \Phi^{S_e}(t)
	~.
\]



\paragraph{Local Algorithms.}
We will consider algorithms that allocate each incoming ball $e$ based on the local information of its current load pair. 
For example, Greedy is local since it allocates each ball to the less loaded incident bin, breaking ties randomly.
In general, let $x^{S_e}_v$ denote the probability that a local algorithm allocates an arrived ball $e$ with load pair $S_e$ to bin $v$.


\begin{theorem}
	\label{thm:potential-based-general}
	Consider a maximization problem with function $f : \Z_{\ge 0} \to \R_{\ge 0}$, and a bound $g : \R_{\ge 0} \to \R_{\ge 0}$.
	The Greedy algorithm guarantees an amortized bound $g$ if there are potential functions $\Phi^{S}$ satisfying:
	\begin{enumerate}
		\item For bin $v$ and its degree $d_v$, $\Phi^{(0,0)}(0) \,\ge\, \dfrac{g(d_v)}{d_v}$.
		\hspace*{\fill} (starting condition)
        \vspace{-6pt}
		\item For any load pair $S$, $\Phi^S(1) = 0$.
		\hspace*{\fill} (terminal condition)
		\item For any bin $v \in V$ and its load $\load_v$, and any load pair $S = (\load_v, \load_u)$:
		\hspace*{\fill} (monotonicity)
		%
		%
        \[
			   f(\load_v + 1) - f(\load_v) + \frac{\dif}{\dif{t}} \Phi^{(\load_v, \load_u)}(t)
			\,\ge\, 
			 2d_v \, \Big( \Phi^{(\load_v, \load_u)}(t) - \Phi^{(\load_v+1, \load_u)}(t) \Big)
			~.
		\]
		\item For any bin $v \in V$ and its load $\load_v$, the change of the potential from increasing $\load_v$, i.e.:
        \[
            \Phi^{(\load_v, \load_u)}(t) - \Phi^{(\load_v+1, \load_u)}(t)
        \]
        is non-decreasing in $\load_u$.
		\hspace*{\fill} (Chebyshev's condition)
	\end{enumerate}
	The theorem also holds for minimization problems, by changing the directions of the inequalities in 1) and 3), and by changing the conclusion of 4) from non-decreasing to non-increasing.
\end{theorem}


\begin{proof}
    We only prove the maximization version by verifying the conditions of \Cref{lem:potential-function-conditions}; the minimization version is almost verbatim.
	The starting and terminal conditions of \Cref{lem:potential-function-conditions} follow from the starting and terminal conditions of the theorem.
	
	It remains to verify the monotonicity condition of \Cref{lem:potential-function-conditions}, i.e.:
    \[
        \frac{\dif}{\dif{t}} \E \big[\Phi(t)+A(t)\big] \ge 0
        ~.
    \]
    We prove a stronger property that the inequality holds \emph{conditioned on what happened before time $t$}, in particular, the bins' current loads.
    
    We first introduce some notations. 
    Let $\loadeset_v$ be the set of load pairs of $v$'s incident balls, i.e.
    $\loadeset_v = \big\{ (\load_v, \load_u)~|~\{u, v\} \in E \big\}$.
    For any $S\in \loadeset_v$, let $\Lambda^{S}_v$ be the normalized total arrival rate of $v$'s incident balls with load pair $S$, i.e.:
    \[
    	\Lambda^{S}_v ~\defeq \frac{1}{2} \sum_{e \sim v \,:\, S_e = S} \lambda_e
    	~.
    \]

    By the definition of $\Lambda^{S}_v$ and the degree of bin $v$, we have:
    \begin{equation}
        \label{eqn:two-chioce-load-pair-rate}
        \sum_{S \in \loadeset_v} \Lambda_v^{S} = d_v
        ~.
    \end{equation}

	At time $t$, the algorithm allocates a ball to bin $v$ with rate $\sum_{S \in \loadeset_v} 2\Lambda^{S}_v x^{S}_v$,
	where $2\Lambda^{S}_v $ is the (unnormalized) arrival rate of balls $e \sim v$ with load sets $S_e = S$, and $x^{S}_v$ is the probability that Greedy allocates such a ball to bin $v$ conditioned on its arrival.
	
	Hence, the algorithm's objective changes by:
	\begin{equation}
		\label{eqn:2choice-algo}
		\frac{\dif{}}{\dif{t}} \E A(t) ~=~ \sum_{v \in V} 
        \sum_{S \in \loadeset_v} 2\Lambda^{S}_v x^{S}_v \, \Big( f(\load_v + 1) - f(\load_v) \Big)
		~.
	\end{equation}

	Further, when the algorithm allocates a ball to bin $v$, its potential decreases accordingly by $\sum_{S \in \loadeset_v} 2\Lambda^{S}_v  \big( \Phi^{S}(t) - \Phi^{S_+}(t) \big)$.
	Summing over $v \in V$, the decrease due to the allocation at time $t$ is:
	\begin{equation}
	    \label{eqn:2choice-pot-allocation}
	    \sum_{v\in V}\,
	    \biggl(\, 
	    	\sum_{S \in \loadeset_v} 2\Lambda^{S}_v x^{S}_v
		\,\biggr)
	    \biggl(\,
	    	\sum_{S \in \loadeset_v} 2\Lambda^{S}_v  \Big( \Phi^{S}(t) - \Phi^{S_+}(t) \Big)
    	\,\biggr)
	    ~.
	\end{equation}

	Finally, the potential also changes due to time elapsing, independent of the arrival and allocation at time $t$, by:
	\begin{equation}
	    \label{eqn:2choice-pot-time}
	    \sum_{e \in E} \lambda_e \cdot\frac{\dif{}}{\dif{t}} \Phi^{S_e}(t)
		\,=\,
		\sum_{ v \in V} \sum_{S \in \loadeset_v} 2\Lambda^{S}_v x^{S}_v \cdot \frac{\dif{}}{\dif{t}} \Phi^{S}(t)
	    ~.
	\end{equation}

	Combining \Cref{eqn:2choice-algo,eqn:2choice-pot-allocation,eqn:2choice-pot-time} and grouping the linear terms w.r.t.\ $\Lambda^S_v$, $\frac{\dif{}}{\dif{t}} \E \big[\Phi(t)+A(t)\big]$ is:
	\[
		\sum_{v \in V} 
		\Bigg( 
		    \sum_{S \in \loadeset_v} 2\Lambda^{S}_v x^{S}_v \,\Big( f(\load_v + 1) - f(\load_v) + \frac{\dif}{\dif{t}} \Phi^{S}(t) \Big)
			\,-\,
			\sum_{S \in \loadeset_v} 2\Lambda^{S}_v x^{S}_v
			\sum_{S \in \loadeset_v} 2\Lambda^{S}_v \Big( \Phi^{S}(t) - \Phi^{S_+}(t) \Big)
		\Bigg)
	    ~.
	\]

	It suffices to show non-negativity for every bin $v \in V$.
    Let $\alpha_S = \frac{\Lambda_v^S}{d_v}$, which sum to $1$ because of Eqn.~\eqref{eqn:two-chioce-load-pair-rate}.
	We can rewrite the desired inequality for bin $v$ as:
	\begin{align*}
		&
		\sum_{S \in \loadeset_v} \alpha_{S} \cdot 2d_v x^{S}_v \Big( f(\load_v + 1) - f(\load_v) + \frac{\dif}{\dif{t}} \Phi^{S}(t) \Big) \\
		& \qquad
		\ge\, 
		\bigg( \sum_{S \in \loadeset_v} \alpha_{S} \cdot 2d_v x^{S}_v \bigg)
		\bigg( \sum_{S \in \loadeset_v} \alpha_{S} \cdot 2d_v \Big( \Phi^{S}(t) - \Phi^{S_+}(t) \Big) \bigg)
		~.
	\end{align*}

    The probability of allocating to bin $v$, i.e., $x_v^{(\load_v, \load_u)}$ is non-decreasing in the other bin's load $\load_u$.
    In the meantime, $\Phi^{(\load_v, \load_u)}(t) - \Phi^{(\load_v+1, \load_u)}(t)$ is also non-decreasing in $\load_u$ by the last condition of the theorem.
	By Chebyshev's sum inequality, the right-hand-side is at most:
	\[
		\sum_{S \in \loadeset_v} \alpha_{S} 
		\cdot 2d_v x^{S}_v 
		\cdot
		2d_v \Big( \Phi^{S}(t) - \Phi^{S_+}(t) \Big)
		~.
	\]
	
	Hence, the inequality reduces to:
    \[
        \sum_{S \in \loadeset_v} \alpha_{S} \cdot 2d_v x^{S}_v \Big( f(\load_v + 1) - f(\load_v) + \frac{\dif}{\dif{t}} \Phi^{S}(t) \Big) 
        ~\ge~
        \sum_{S \in \loadeset_v} \alpha_{S} 
		\cdot 2d_v x^{S}_v 
		\cdot
		2d_v \Big( \Phi^{S}(t) - \Phi^{S_+}(t) \Big)
        ~,
    \]
    which is a weighted combination of the third condition of the theorem.
\end{proof}

\subsection{Applications}

Next, we consider several classes of functions $f$ to demonstrate the applications of \Cref{thm:potential-based-general}.
We focus on instances in which the degrees of bins are bounded by some positive number $d$, i.e., $d_v \le d$ for all bins $v \in V$.

\begin{shaded}
\vspace{-6pt}
\begin{example}[Concave Function]
    A function $f : \Z_{\ge 0} \to \R_{\ge 0}$ is \emph{concave} if for any $\load \in \Z_{\ge 0}$:
    \[
      f(\load+2) - f(\load+1) \le f(\load+1) - f(\load)
      ~.
    \]
    Observe that such a concave function with non-negative range must be \emph{non-decreasing}.
    Finally, we say that it is \emph{normalized} if $f(0) = 0$.
\end{example}
\vspace{-6pt}
\end{shaded}

\begin{corollary}
    \label{cor:concave-amortized}
    For any normalized concave function $f : \Z_{\ge 0} \to \R_{\ge 0}$, and any instance with degree upper bound $d$, 
    Greedy guarantees that the expected objective is at least:
    \begin{equation*}
        \E \sum_{v \in V}f(\load_v) ~\ge~
        \E_{j \sim \mathrm{Poisson}(2d)} \frac{f( \lfloor \nicefrac{j}{2} \rfloor ) + f( \lceil \nicefrac{j}{2} \rceil )}{2d}
        \cdot \sum_{v \in V} d_v ~.
    \end{equation*}
\end{corollary}

In fact, this amortized guarantee is optimal for $g(d_v = d)$, because no algorithm can achieve a better bound in the simple instance of two bins, and one ball with arrival rate $2d$ between them.

We remark that Greedy algorithm guarantees amortized optimality for maximizing general (not normalized) concave functions.
This can be seen by decomposing $f$ into the sum of a normalized concave function and a constant,
with the latter's contribution being independent of the algorithm's decision.

\begin{shaded}
\vspace{-6pt}
\begin{example}[Matching]
    Recall that the maximization problem with function $f(\load) = \min \{ \load, 1 \}$ corresponds to a matching between balls and bins.
    More generally, function $f (\load) = \min \{ \load, c \}$ for some positive integer capacity $c$ corresponds to $c$-matching.
\end{example}
\vspace{-6pt}
\end{shaded}

On the one hand, the $c$-matching functions are normalized and concave.
Hence, it is a special case of the previous example.
We prove the following result by constructing potential functions directly. See \Cref{app:c-matching}.

\begin{corollary}
    \label{cor:c-matching}
    For the $c$-matching problem and any instance with degree upper bound $d$,
    Greedy guarantees that the expected number of balls allocated to bins, each capped by $c$, is at least:
    \begin{equation*}
        \E_{j \sim \mathrm{Poisson}(2d)} \frac{\min \{ j, 2c \}}{2d}
        \cdot \sum_{v \in V} d_v~.
    \end{equation*}
\end{corollary}

On the other hand, any normalized concave function $f$ can be decomposed into a linear combination of the $c$-matching functions $\min\{L, c\}$.
We state this as the next lemma and defer its proof to \Cref{app:concave-decompose}.

\begin{lemma}
    \label{lem:concave-decompose}
    A normalized concave function $f : \Z_{\ge 0} \to \R_{\ge 0}$ can be decomposed into:
    \begin{equation*}
        f(\load) ~=~ \sum_{c = 1}^\infty \Big( 2 f(c) - f(c-1) - f(c+1) \Big)  \cdot \min\big\{ \load, c \big\}
        ~.
    \end{equation*}
\end{lemma}

In this sense, the $c$-matching functions form a basis of general concave (and as we will see shortly, convex) functions.
The bound for normalized concave functions (\Cref{cor:concave-amortized}) can be proved via this decomposition, and we defer the proof to \Cref{app:concave-amortized}.

%
%



\begin{shaded}
\vspace{-6pt}
\begin{example}[Convex Function]
    A function $f : \Z_{\ge 0} \to \R_{\ge 0}$ is \emph{concave} if for any $\load \in \Z_{\ge 0}$:
    \[
      f(\load+2) - f(\load+1) \ge f(\load+1) - f(\load)
      ~.
    \]
    Finally, we say that it is \emph{normalized} if $f(0) = f(1) = 0$.
\end{example}
\vspace{-6pt}
\end{shaded}

\begin{corollary}
    \label{cor:convex-amortized}
    For any normalized convex function $f : \Z_{\ge 0} \to \R_{\ge 0}$ and any instance with degree upper bound $d$, 
    Greedy guarantees that the expected objective is at most:
    \begin{equation*}
        \E \sum_{v \in V}f(\load_v) ~\le~
        \E_{j \sim \mathrm{Poisson}(2d)} \frac{f( \lfloor \nicefrac{j}{2} \rfloor ) + f( \lceil \nicefrac{j}{2} \rceil )}{2d}
        \cdot \sum_{v \in V} d_v ~.
    \end{equation*}
\end{corollary}

Similar to the case of concave functions, this amortized guarantee is optimal for $g(d_v = d)$, because no algorithm can achieve a better bound in the simple instance of two bins, and one ball with arrival rate $2d$ between them.

We also remark that Greedy guarantees amortized optimality for minimizing general convex functions, because
any convex function can be decomposed into the sum of a normalized convex function and a linear function,
with the latter's contribution being independent of the algorithm's decision.

The analysis for normalized convex functions also reduces to the $c$-matching case, as these functions can be decomposed into the difference between a linear function and $c$-matching functions.
We state this as the next lemma.

\begin{lemma}
	\label{lem:convex-decomposition}
    A normalized convex function $f : \Z_{\ge 0} \to \R_{\ge 0}$ can be decomposed into:
    \begin{equation*}
        f(\load) ~=~ \sum_{c=1}^\infty \Big( f(c-1) + f(c+1) - 2 f(c) \Big) \cdot \big( L - c \big)^+
        ~,
    \end{equation*}
    where note that $\big(\load - c\big)^+ = \load - \min \big\{ \load, c \big\}~.$
\end{lemma}

Finally, we examine two special classes of convex functions.

\begin{shaded}
\vspace{-6pt}
\begin{example}[Load Balancing]
    Recall that load balancing corresponds to function $f(\load) = e^{\eta \load}$ through the softmax function $\frac{1}{\eta} \ln \sum_{v \in V} e^{\eta \load_v}$.
    We consider $\eta = 1$ for ease of demonstration.
\end{example}
\vspace{-6pt}
\end{shaded}


\begin{corollary}
\label{cor:load-balancing}
    For the load balancing problem with $f(\load) = e^{\load}$ and any instance with degree upper bound $d$, 
    Greedy guarantees that the expected objective is at most:
    \begin{equation*}
        \E \sum_{v \in V}f(\load_v) ~\le~
        \E_{j \sim \mathrm{Poisson}(2d)} \frac{e^{ \lfloor \nicefrac{j}{2} \rfloor } + e^{ \lceil \nicefrac{j}{2} \rceil } - (e-1)(j-2d) - 2}{2d}
        \cdot \sum_{v \in V} d_v + |V|~.
    \end{equation*}
\end{corollary}

\begin{shaded}
\vspace{-6pt}
\begin{example}[Completion Time Minimization]
    Recall that completion time minimization corresponds to function $f(\load) = \frac{\load^2 + \load}{2}$.
\end{example}
\vspace{-6pt}
\end{shaded}

\begin{corollary}
\label{cor:completion-time-min}
    For the completion time minimization problem and any instance with degree upper bound $d$,
    Greedy guarantees that the expected total completion time is at most:
    \begin{equation*}
        \E \sum_{v \in V}f(\load_v) ~\le~
        \E_{j \sim \mathrm{Poisson}(2d)} \frac{{\lfloor \nicefrac{j}{2} \rfloor}^2 + {\lceil \nicefrac{j}{2} \rceil}^2 - j + 4d}{4d}
        \cdot \sum_{v \in V} d_v ~.
    \end{equation*}
\end{corollary}

\section{Mark-and-Swap Algorithm}
\label{sec:mark-swap}

This section presents another algorithm called Mark-and-Swap for the Generalized Balls into Bins problem.
Compared to the amortized bounds by Greedy, the bounds by Mark-and-Swap are not as tight but still better than the baseline by Random, and importantly, hold \emph{without amortization}.
Moreover, the analysis of Mark-and-Swap can be extended to handle combinatorial functions related to edge-weighted matching and completion minimization.
\Cref{sec:completion-time} will demonstrate the case of completion time minimization.

\subsection[1-Regular Graphs]{$1$-Regular Graphs}

We first present and analyze the main component of the Mark-and-Swap algorithm, by considering the special case of $1$-regular instances, i.e., when $d_v = 1$ for all bins $v \in V$.

\begin{algorithm}{Mark-and-Swap Algorithm for $1$-Regular Instances}
    When a ball $\{u, v\}$ arrives:
    \begin{enumerate}
        \item Sample bin $w \in \{u, v\}$ uniformly at random, which we will refer to as the \emph{first choice}.
        \item Allocate the ball to bin $w$ if it is not marked.
        \item Otherwise, allocate it to the other bin, which we will refer to as the \emph{second choice}.
        \item Mark bin $w$.
    \end{enumerate}
\end{algorithm}

We now characterize the stochastic process by which Mark-and-Swap allocates balls into bins.

\begin{definition}[Directed Balls]
	If a ball $\{u, v\}$ arrives and samples $w = u$ in the first step, we say it is a directed ball $(u, v)$. 
	Likewise, if it samples $w = v$, we say it is a directed ball $(v, u)$.
\end{definition}


Since $w$ distributes uniformly over $\{u, v\}$ with fresh randomness for each ball, the directed balls arrive by a Poisson process.

\begin{lemma}
	\label{lem:directed-ball-rate}
	For any $u \ne v \in V$, directed balls $(u,v)$ arrive by a Poisson process with rate $\frac{\lambda_{\{u, v\}}}{2}$.
\end{lemma}

By definition, the algorithm marks a bin $v$ whenever a directed ball arrives with $v$ as the first choice.
The total arrival rate of these sub-balls is $\sum_{u \ne v} \frac{\lambda_{\{u, v\}}}{2} = d_v = 1$.
As a result, we have:

\begin{lemma}
	\label{lem:regular-marking}
	For any bin $v \in V$, the events of marking bin $v$ follow a Poisson process with rate $1$.
\end{lemma}

This lemma yields two further lemmas as corollaries.
First, let random variable $F_v$ be the number of balls bin $v$ receives as the first choice.
By the definition of Mark-and-Swap, bin $v$ receives at most one such ball, after which it is marked and can only receive balls as their second choice.
Further, bin $v$ receives such a ball if and only if it is marked at least once.
Hence, we get the distribution of $F_v$ as a corollary of \Cref{lem:regular-marking}.

\begin{lemma}
	\label{lem:regular-first-choice}
	For any bin $v \in V$, $F_v \sim \Bernoulli(1-\frac{1}{e})$.
\end{lemma}

Further, let random variable $\MarkTime_v$ denote the first time when bin $v$ is marked. 
Define $\MarkTime_v = 1$ if bin $v$ is never marked. 
The distribution of $\MarkTime_v$ also follows as a corollary of \Cref{lem:regular-marking}.

\begin{lemma}
	\label{lem:regular-first-mark}
	For any bin $v \in V$,  $\MarkTime_v$ follows the distribution defined by:
	\[
		\Pr \big[ \MarkTime_v \ge t \big] =
		\begin{cases}
			e^{-t} & \mbox{if $0 \le t \le 1$;} \\
			0 & \mbox{otherwise.}
		\end{cases}
	\]
	That is, $\MarkTime_v$ follows the exponential distribution with unit rate, except that its value is capped by $1$.
\end{lemma}

Next, we analyze the number of balls allocated to bin $v \in V$ with $v$ as the second choice.
For any $u \ne v$, let $S_{(u,v)}$ be the number of directed balls $(u, v)$ allocated to their second choice bin $v$.
By \Cref{lem:directed-ball-rate}, these balls arrive with rate $\frac{\lambda_{\{u, v\}}}{2}$.
Further, Mark-and-Swap allocates these balls to bin $v$ after bin $u$ is marked at time $\MarkTime_u$, which follows the distribution defined in \Cref{lem:regular-first-mark}.

Let us introduce a notation for such distributions, because we will repeatedly refer to them in the rest of the section.

\begin{definition}[Distribution of Second-choice Allocations]
	\label{def:swap}
	For any $\lambda \ge 0$, let $\Swap(\lambda)$ be the distribution of a random variable $S$ realized as follows:
	\begin{enumerate}
	    \item Sample $T \sim \Exponential(1)$.
	    \item Sample $S \sim \Poisson\bigl( \lambda (1-T)^+ \bigr)$.
	\end{enumerate}
\end{definition}

We can then summarize the above discussion as the following lemma.

\begin{lemma}
	\label{lem:regular-second-choice}
	For any $u \ne v \in V$, $S_{(u, v)} \sim \Swap\big( \frac{\lambda_{\{u,v\}}}{2} \big)$.	
\end{lemma}

%



\Cref{lem:regular-first-choice,lem:regular-second-choice} together fully characterize the distribution of bin $v$'s load $\load_v$.
Recall that $D_1 + D_2$ denotes the distribution of $X_1 + X_2$ where $X_1 \sim D_1$ and $X_2 \sim D_2$ independently.

\begin{lemma}
    \label{lem:regular-load}
    For any bin $v \in V$, we have $\load_v \sim \Bernoulli(1 - \frac{1}{e}) + \sum_{u \ne v} \Swap(\frac{\lambda_{\{u, v\}}}{2})$.
\end{lemma}


To better understand the distribution of loads, we next show several properties of distribution $\Swap(\lambda)$, which will give the worst-case distribution of load $\load_v$ and be useful for analyzing the Mark-and-Swap algorithm in the general case.
Recall that $D_1 \succeq D_2$ means $D_1$ (second-order) stochastically dominates $D_2$.

\begin{lemma}[``Concavity'' of $\Swap(\lambda)$] 
    \label{lem:swap-concave}
    For any $0 \le \varepsilon \le \alpha \le \beta$, we have:
    \[
    	\Swap(\alpha) + \Swap(\beta) \,\succeq\, \Swap(\alpha - \varepsilon) + \Swap(\beta + \varepsilon)
    	~.
    \]
\end{lemma}

We defer the proof to \Cref{app:swap-concave}.
Note that distribution $\Swap(0)$ is a point mass at $0$. 
By letting $\varepsilon = \alpha$, we have:

\begin{corollary}
	\label{cor:swap-merge}
    For any $\alpha, \beta \ge 0$:
    \[
    	\Swap(\alpha) + \Swap(\beta) 
    	\,\succeq\,
    	\Swap(\alpha + \beta)
    	~.
    \]
\end{corollary}

Combining \Cref{cor:swap-merge} with \Cref{lem:regular-load} and recalling that $\sum_{u \ne v} \frac{\lambda_{\{u, v\}}}{2} = d_v = 1$, we get the worst-case distribution of the loads.

\begin{theorem}
    \label{thm:mark-swap-regular-worst-case}
    For any $1$-regular instance and any bin $v \in V$, Mark-and-Swap ensures that $\load_v$ stochastically dominates $\Bernoulli(1 - \frac{1}{e}) + \Swap(1)$.
\end{theorem}

We make two remarks about the above worst-case distribution of loads by Mark-and-Swap.
First, this distribution can be achieved in the simple example with two bins and a ball between them with an arrival rate of $2$.
Second, it stochastically dominates the baseline distribution $\Poisson(1)$ given by the Random algorithm.
For the latter claim, we will prove a more general lemma that will be useful for our analysis of the general case in the next subsection.

\begin{lemma}
	\label{lem:markswap-dominate-poisson}
	For any $0 \le \lambda \le 1$:
	\[
		\Bernoulli \Big(\big(1 - \frac{1}{e}\big) \lambda \Big) + \Swap(\lambda)
		\,\succeq\,
		\Poisson(\lambda)
		~.
	\]
\end{lemma}

We defer the proof to \Cref{app:markswap-dominate-poisson}.
Finally, we give the mean and variance of distribution $\Swap(\lambda)$.
They follow by the definition of the distribution and basic calculus, which we omit.
	
	


\begin{lemma}
	\label{lem:swap-expectation}
	For any $\lambda \ge 0$, $\E \Swap(\lambda) = \frac{\lambda}{e}$.	
\end{lemma}

With $\lambda = 1$, it implies that the expected load of a bin is $\E \big[ \Bernoulli(1-\frac{1}{e}) + \Swap(1) \big] = 1$, which equals the degree of the bin.
In this sense, Mark-and-Swap is mean-preserving.
This will also be true in the general case, which we will show in the next subsection.


\begin{lemma}
    \label{lem:swap-variance}
    For any $\lambda \ge 0$, $\Var \Swap(\lambda) = e ^ {-1} \lambda + (1 - 2e ^ {-1} - e ^ {-2}) \lambda ^ 2$.
\end{lemma}
%


We conclude the subsection by applying \Cref{thm:mark-swap-regular-worst-case} to the examples of matching and completion time minimization.


\begin{shaded}
\vspace{-6pt}
\begin{example}[Matching]
	Consider function $f(L) = \min \{L, 1\}$.
	We have:
	\[
		\E_{L \sim \Bernoulli(1 - \frac{1}{e}) + \Swap(1)} \min \{L, 1\} = 1 - \frac{2}{e ^ 2}
		~,
	\]	
	since $\Pr_{F \sim \Bernoulli(1-\frac{1}{e})} [F = 0] = \frac{1}{e}$, and $\Pr_{S \sim \Swap(1)}[S = 0] = \frac{2}{e}$.
	This matches the optimal amortized bound by Greedy.
	By contrast, the baseline distribution only gives:
	\[
		\E_{L \sim \Poisson(1)} \min \{L, 1\} = 1 - \frac{1}{e}
		~.
	\]	
\end{example}

\begin{example}[Completion Time Minimization]
	Consider function $f(L) = \frac{L^2 + L}{2}$.
	We have:
    \[
        \E_{L \sim \Bernoulli(1-\frac{1}{e}) + \Swap(1)} \frac{L^2 + L}{2} = \frac{3}{2} - \frac{1}{e ^ {2}} \approx 1.36
        ~,
    \]
    by \Cref{lem:swap-expectation,lem:swap-variance} and properties of the Bernoulli distribution.
    By contrast, the optimal amortized bound by Greedy is:
    \[
    	\E_{L \sim \Poisson(2)} \frac{1}{2} \left( \frac{\lfloor \frac{L}{2}  \rfloor ^ 2 + \lfloor \frac{L}{2}  \rfloor}{2} + \frac{\lceil \frac{L}{2}  \rceil ^ 2 + \lceil \frac{L}{2}  \rceil}{2} \right)= \frac{21}{16} - \frac{1}{16e^4} \approx 1.31
        ~,
    \]
    by \Cref{cor:convex-amortized},
    and the baseline distribution only gives:
	\[
		\E_{L \sim \Poisson(1)} \frac{L^2 + L}{2} = \frac{3}{2}
		~.
	\]	
\end{example}
\vspace{-6pt}
\end{shaded} 

%

\subsection{General Case}

To define Mark-and-Swap for general instances, we reduce the problem to the $1$-regular case by splitting the balls and bins into sub-balls and sub-bins such that every sub-bin's degree effectively equals $1$.
Let us first define sub-balls and sub-bins formally.

\begin{definition}[Sub-bins]
	For each bin $v \in V$, let there be $\lceil d_v \rceil$ unit-size \emph{sub-bins}, denoted as $v(i)$ for integers $1 \le i \le \lceil d_v \rceil$. 
\end{definition}

\begin{definition}[Sub-balls]
	For any ball $\{u, v\}$ and any integers $i, j \ge 1$, a sub-ball $\{ u(i), v(j) \}$ is incident to sub-bins $u(i)$ and $v(j)$.
\end{definition}



Next, we present the sampling procedure that defines a randomized mapping from balls $\{u, v\}$ to sub-balls $\{ u(i), v(j) \}$.
See \Cref{fig:mapping} for an illustration.

\begin{algorithm}{Mark-and-Swap Algorithm: Sampling Sub-balls}
    \begin{itemize}

        \item Let $\succ_v$ be an arbitrary total order over the edges incident to $v$.
        \item For each arrived ball $\{ u, v \}$:
        \begin{itemize}
            \item Sample $\theta \in (0, \lambda_{uv}]$ uniformly at random.
            \item Let $i = \lceil \sum_{e \sim u \,:\, e \,\succ_u (u, v)} \lambda_e + \theta \rceil$.
            \item Let $j = \lceil \sum_{e \sim v \,:\, e \,\succ_v (u, v)} \lambda_e + \theta \rceil$.
            \item Let it be a sub-ball $\{ u(i), v(j) \}$.
        \end{itemize}
    \end{itemize}
\end{algorithm}

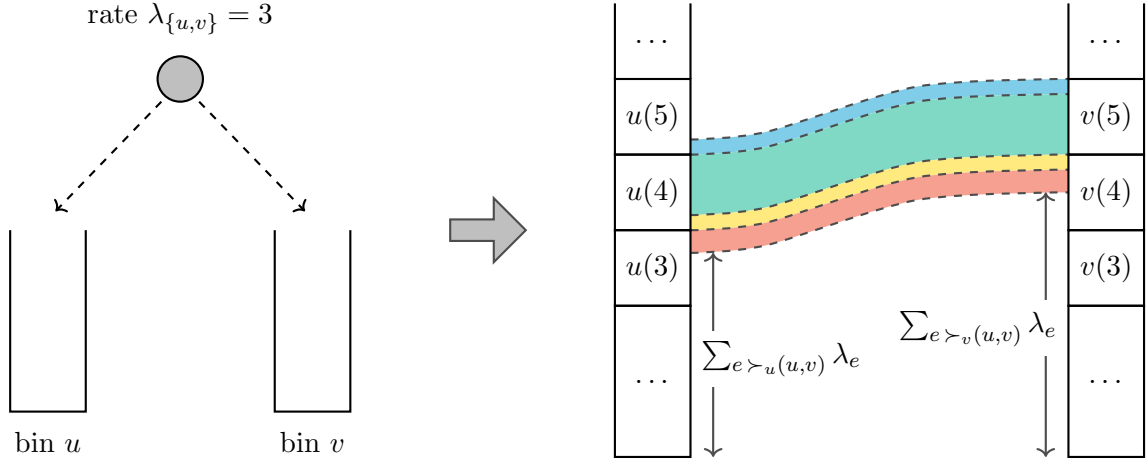
\begin{figure}[h]
	\centering
	\begin{tikzpicture}	
		\draw[draw=black, thick] (-8,2) -- (-8,-.4) -- (-7,-.4) -- (-7,2);
		\node at (-7.5,-.8) {bin $u$};
		\draw[draw=black, thick] (-4.5,2) -- (-4.5,-.4) -- (-3.5,-.4) -- (-3.5, 2);
		\node at (-4,-.8) {bin $v$};	
		\draw[draw=black, fill=gray!50, thick] (-5.75,4) circle (.3);
		\draw[draw=black, thick, dashed, ->] (-6,3.7) -- (-7.4,2.25);
		\draw[draw=black, thick, dashed, ->] (-5.5,3.7) -- (-4.1,2.25);	
		\node at (-5.75,4.8) {rate $\lambda_{\{u,v\}} = 3$};
		\node[
			single arrow, 
			draw=black!70, 
			fill=gray!50, 
			thick,
      		single arrow head extend=8pt,
      		minimum height=10mm
      	] at (-1.75,2) {};			
		\draw[draw=none, thick, dashed, name path=A] plot[smooth] coordinates { (1,3.2) (2, 3.3) (4,3.9) (6,4) };
		\draw[draw=none, thick, dashed, name path=B] plot[smooth] coordinates { (1,3) (2, 3.1) (4,3.7) (6,3.8) };
		\draw[draw=none, thick, dashed, name path=C] plot[smooth] coordinates { (1,2.2) (2, 2.3) (4,2.9) (6,3) };
		\draw[draw=none, thick, dashed, name path=D] plot[smooth] coordinates { (1,2) (2, 2.1) (4,2.7) (6,2.8) };
		\draw[draw=none, thick, dashed, name path=E] plot[smooth] coordinates { (1,1.7) (2, 1.8) (4,2.4) (6,2.5) };		
		\tikzfillbetween[of=A and B, on layer=bg]{hkublue!50};
		\tikzfillbetween[of=B and C, on layer=bg]{hkugreen!50};
		\tikzfillbetween[of=C and D, on layer=bg]{hkuyellow!50};
		\tikzfillbetween[of=D and E, on layer=bg]{hkured!50};
		\draw[draw=black!70, thick, dashed] plot[smooth] coordinates { (1,3.2) (2, 3.3) (4,3.9) (6,4) };
		\draw[draw=black!70, thick, dashed] plot[smooth] coordinates { (1,3) (2, 3.1) (4,3.7) (6,3.8) };
		\draw[draw=black!70, thick, dashed] plot[smooth] coordinates { (1,2.2) (2, 2.3) (4,2.9) (6,3) };
		\draw[draw=black!70, thick, dashed] plot[smooth] coordinates { (1,2) (2, 2.1) (4,2.7) (6,2.8) };
		\draw[draw=black!70, thick, dashed] plot[smooth] coordinates { (1,1.7) (2, 1.8) (4,2.4) (6,2.5) };
		\draw[draw=black, thick] (0,-1) rectangle (1,1);
		\draw[draw=black, thick] (0,1) rectangle (1,2);
		\draw[draw=black, thick] (0,2) rectangle (1,3);
		\draw[draw=black, thick] (0,3) rectangle (1,4);
		\draw[draw=black, thick] (0,4) -- (0,5);
		\draw[draw=black, thick] (1,4) -- (1,5);
		\node at (0.52,0) {$\dots$};
		\node at (0.5,1.5) {$u(3)$};
		\node at (0.5,2.5) {$u(4)$};
		\node at (0.5,3.5) {$u(5)$};				
		\node at (0.52,4.5) {$\dots$};
		\draw[draw=black, thick] (6,-1) rectangle (7,1);
		\draw[draw=black, thick] (6,1) rectangle (7,2);
		\draw[draw=black, thick] (6,2) rectangle (7,3);
		\draw[draw=black, thick] (6,3) rectangle (7,4);
		\draw[draw=black, thick] (6,4) -- (6,5);
		\draw[draw=black, thick] (7,4) -- (7,5);
		\node at (6.52,0) {$\dots$};
		\node at (6.5,1.5) {$v(3)$};
		\node at (6.5,2.5) {$v(4)$};
		\node at (6.5,3.5) {$v(5)$};		
		\node at (6.52,4.5) {$\dots$};		
		\draw[draw=black!70, thick, <->] (1.3,-1) -- (1.3,1.7);
		\node[fill=white] at (2.2,.3) {$\sum_{e \,\succ_u (u, v)} \lambda_e$};
		\draw[draw=black!70, thick, <->] (5.7,-1) -- (5.7,2.5);
		\node[fill=white] at (4.8,.7) {$\sum_{e \,\succ_v (u, v)} \lambda_e$};			
	\end{tikzpicture}	
	\caption{%
		An illustration of the mapping from balls to sub-balls. 
		In the example, the balls between bins $u$ and $v$ arrive at a rate of $3$. 
		They are mapped to $4$ sub-balls \textcolor{hkublue}{$(u(5), v(5))$}, \textcolor{hkugreen}{$(u(4), v(5))$}, \textcolor{hkuyellow}{$(u(4), v(4))$}, and \textcolor{hkured}{$(u(3), v(4))$}, represented by the curved stripes in $4$ colors.
		The height of each stripe is the normalized (divided by $2$) arrival rate of the sub-ball.
	}
	\label{fig:mapping}
\end{figure}

By definition, we have the following facts.

\begin{lemma}
	\label{lem:sub-ball-poisson}
	The sub-balls arrive independently by a Poisson process.	
\end{lemma}

Let $\lambda_{\{ u(i), v(j) \}}$ be the arrival rate of sub-ball $\{ u(i), v(j) \}$.
We only need the following relation between the arrival rates of the sub-balls and those of the original balls.
The precise dependence is unimportant.

\begin{lemma}
	\label{lem:sub-ball-total-rate}
    For any balls $\{u, v\} \in E$, $\lambda_{\{ u, v \}} = \sum_{i = 1} ^ {\lceil d_u \rceil } \sum_{j = 1} ^ { \lceil d_v \rceil} \lambda_{\{ u(i), v(j) \}}$.
\end{lemma}

\begin{lemma}
	\label{lem:sub-bin-degree}
    For any ball $v \in V$:
    \[
        d_{v(i)} = 
        \begin{cases}
        	~~1 & 1 \le i \le \lfloor d_v \rfloor ; \\
        	\{ d_v \} & i = \lceil d_v \rceil > \lfloor d_v \rfloor \mbox{ (if applicable)} .
        \end{cases}
    \]
\end{lemma}


If all sub-bins have degrees $1$, which would be the case if the bins have integer degrees, then we can treat the sub-balls and sub-bins as balls and bins of a $1$-regular instance and apply the algorithm from the last subsection.
Note that we need to sum up the loads of sub-bins $v(i)$, $i \ge 1$, for every bin $v$, before applying the function $f$.
Hence, the analysis is non-trivial even in this case.

If some bin $v$ has a non-integer degree, then sub-bin $v(\lceil d_v \rceil)$ would have degree $\{ d_v \} < 1$.
In this case, we artificially introduce another Poisson process with rate $1-\{d_v\}$, upon whose arrival we also mark sub-bin $v(\lceil d_v \rceil)$.
This may be seen as introducing a pseudo-sub-ball whose sole purpose is to increase the rate of marking the sub-bin to $1$, i.e., to restore \Cref{lem:regular-marking}.

\begin{algorithm}{Mark-and-Swap Algorithm: Online Allocation}
    \begin{itemize}
        \item When a sub-ball $\{ u(i), v(j) \}$ arrives:
        \begin{enumerate}
            \item Sample $w \in \{u(i), v(j)\}$ uniformly at random.
            \item Assign the sub-ball to sub-bin $w$ if it is not marked.
            \item Otherwise, assign it to the other sub-bin.
            \item Mark sub-bin $w$.
        \end{enumerate}
        \item For every bin $v$, further mark sub-bin $v(\lceil d_v \rceil)$ by a Poisson process with rate $1 - \{ d_v \}$.
    \end{itemize}
\end{algorithm}


%
%




Recall the main result of the $1$-regular case (\Cref{thm:mark-swap-regular-worst-case}) that the load of a $1$-regular bin stochastically dominates distribution $\Bernoulli(1-\frac{1}{e}) + \Swap(1)$.
Now that the last sub-bin of every bin may have a degree less than $1$, we need to consider distributions with more flexible parameters.

\begin{definition}
    \label{def:markswap}
    For any $\lambda \in [0, 1]$, define the distribution $\MarkSwap(\lambda)$ as:
    \[
        \MarkSwap(\lambda) \,\defeq\, \Bernoulli \Big( \big(1 - \frac{1}{e} \big) \lambda \Big) + \Swap(\lambda)
        ~.
    \]    
  	Further, extend the definition to $\lambda > 1$ by letting:
  	\[
  		\MarkSwap(\lambda) \defeq \lfloor \lambda \rfloor \cdot \MarkSwap(1) + \MarkSwap(\{ d_v \})
  		~.
  	\]
\end{definition}

As a corollary of the above definition and \Cref{lem:markswap-dominate-poisson}, we have:
\begin{lemma}
	\label{lem:mark-swap-worst-case}
	For any $\lambda \ge 0$, $\MarkSwap(\lambda) \succeq \Poisson(\lambda)$.
\end{lemma}


\begin{theorem}
    \label{thm:mark-swap-worst-case}
    For any instance and any bin $v \in V$, the Mark-and-Swap algorithm guarantees that $\load_v \succeq \MarkSwap(d_v)$.
\end{theorem}

Since $\MarkSwap(\lambda) \succeq \Poisson(\lambda)$ for $0 \le \lambda \le 1$ (\Cref{lem:markswap-dominate-poisson}) and the sum of Poisson distributions is still a Poisson distribution with the rates summed up, the worst-case distribution of loads in \Cref{thm:mark-swap-worst-case} stochastically dominates the baseline distribution $\Poisson(d_v)$.

It is tempting to think that the theorem follows directly by having $\lfloor d_v \rfloor$ sub-bins whose worst-case load distribution is $\MarkSwap(1)$, and one last sub-bin whose worst-case load distribution is $\MarkSwap(\{d_v\})$.
However, this argument is flawed because the loads of different sub-bins may be correlated.
For example, if there were sub-balls between sub-bin $u(1)$ and both $v(1)$ and $v(2)$.
Then, marking $u(1)$ early would lead to greater loads for both $v(1)$ and $v(2)$.

The correct proof below accounts for the first-choice contribution to the total load of a bin $v$ by each sub-bin $v(i)$, and the second-choice contribution to bin $v$ by the first-choices $u(j)$ of the contribution.
This restores the independence of the random variables involved in the proof.

\begin{proof}[Proof of \Cref{thm:mark-swap-worst-case}]
	Let $F_{v(i)}$ be the number of balls that sub-bin $v(i)$ receives as the first choice:
 	\[
 		F_{v(i)} \sim 
 		\begin{cases}
 			\Bernoulli \big( 1-\frac{1}{e} \big) & 1 \le i \le \lfloor d_v \rfloor; \\[2ex]
 			\Bernoulli \big( ( 1-\frac{1}{e} ) \{ d_v \} \big) & i = \lceil d_v \rceil > \lfloor d_v \rfloor \mbox{ (if applicable)}.
 		\end{cases}
 	\]
 	
 	The first case follows by \Cref{lem:regular-first-choice}, while the second case holds because the event of marking sub-bin $v(\lceil d_v \rceil)$ for the first time is an actual allocation with probability $\{d_v\}$.
 	
 	Further, for every sub-bin $u(j)$ for $u \ne v$ and $1 \le j \le \lceil d_u \rceil$, let $S_{(u(j),v(i))}$ be the number of directed balls $(u(j), v(i))$ allocated to their second choice bin $v(i)$.
 	Let the total number of such second-choice allocations from $u(j)$ be:
 	\[
 		S_{u(j)} \,\defeq\, \sum_{i=1}^{\lceil d_v \rceil} S_{(u(j),v(i))}
 		~.
 	\]
 	
 	Correspondingly, let the normalized total rate of sub-balls between $u(j)$ and sub-bins $v(i)$ be:
 	\[
 		\mu_{u(j)} \,\defeq\, \frac{1}{2} \sum_{i=1}^{\lceil d_v \rceil} \lambda_{\{u(j),v(i)\}}
 	\]
 	
 	We have $S_{u(j)} \sim \Swap ( \mu_{u(j)} )$.
 	Combining with the independence of $F_{v(i)}$'s and $S_{u(j)}$'s, we conclude that bin $v$'s load $\load_v = \sum_{i=1}^{\lceil d_v \rceil} F_{v(i)} + \sum_{u \ne v} \sum_{j=1}^{\lceil d_u \rceil} S_{u(j)}$ follows distribution:
 	\[
        \lfloor d_v \rfloor \cdot \Bernoulli\big( 1 - \frac{1}{e} \big) + \Bernoulli \Big( \big(1 - \frac{1}{e}\big) \{ d_v \} \Big) + \sum_{u \ne v} \sum_{j=1}^{\lceil d_u \rceil} \Swap(\mu_{u(j)})
        ~.
    \]
    
    Finally, repeatedly applying \Cref{lem:swap-concave} and \Cref{cor:swap-merge} shows that:
    \[
    	\sum_{u \ne v} \sum_{j=1}^{\lceil d_u \rceil} \Swap(\mu_{u(j)})\succeq \lfloor d_v \rfloor \cdot \Swap (1) + \Swap \big \{d_v\} \big)
    	~,
    \]
    where we use the fact that $0 \le \mu_{u(j)} \le 1$ and $\sum_u \sum_j \mu_{u(j)} = d_v$.
	The theorem then follows.
\end{proof}


Next, we present the mean and variance of distribution $\MarkSwap(\lambda)$. 
They follow by the definition of distribution, \Cref{lem:swap-expectation}, \Cref{lem:swap-variance}, and properties of Bernoulli distribution.
\begin{lemma}
    \label{lem:markswap-expectation}
    For any $\lambda \ge 0$, $\E \MarkSwap(\lambda) = \lambda$.
\end{lemma}

Combining \Cref{thm:mark-swap-worst-case} and \Cref{lem:markswap-expectation} shows that the expected load of a bin $v \in V$ equals its degree $d_v$, as promised in a discussion in the last subsection.
Moreover, the variance is strictly smaller than the counterpart by the baseline distribution $\Poisson(\lambda)$. 

\begin{lemma}
    \label{lem:markswap-variance}
    For any $0 \le \lambda \le 1$, $\Var \MarkSwap(\lambda) = \lambda - \frac{2}{e^2} \lambda^2$.
\end{lemma}

\begin{shaded}
\vspace{-6pt}
\begin{example}[Completion Time Minimization]
	Consider function $f(L) = \frac{L^2 + L}{2}$.
	We have:
    \[
        \E_{L \sim \MarkSwap(d_v)} \frac{L^2 + L}{2} ~=~ \frac{1}{2}d_v ^ 2 + d_v - e ^ {-2} \lfloor d_v \rfloor - e ^ {-2} \{ d_v \} ^ 2
        ~,
    \]
    by \Cref{lem:markswap-expectation,lem:markswap-variance}.
    By contrast, the baseline distribution only gives:
	\[
		\E_{L \sim \Poisson(d_v)} \frac{L^2 + L}{2} ~=~ \frac{1}{2} d_v ^ 2 + d_v
		~.
	\]	
\end{example}
\vspace{-6pt}
\end{shaded}

\section{Applications in Completion Time Minimization}
\label{sec:completion-time}

\paragraph{Online Stochastic Completion Time Minimization.}
Consider a set of machines $M$ and a set of (types of) jobs $J$.
We will refer to a job of type $j$ simply as a job $j$ for brevity.
For any $j \in J$, jobs $j$ arrive by a Poisson process with rate $\lambda_j$ in time horizon $[0, 1]$.
Processing a job $j$ on machine $i$ takes time $s_{ij}$, which we will refer to as the \emph{size} of job $j$ on machine $i$.

When a job arrives, the algorithm observes its type and immediately allocates it to a machine.
After all jobs are allocated, each machine processes the allocated jobs by Smith's Rule, i.e., in ascending order of their sizes. 
A job's \emph{completion time} is the sum of the sizes of all jobs processed before it and its own size.
We want to minimize the expected total completion time of the jobs.

\bigskip

Next, we design an online algorithm for this problem as an application of the Balls into Bins algorithm from \Cref{sec:mark-swap}.


\begin{theorem}
	\label{thm:completion-time}
	There is a $1.435$-competitive online algorithm for Online Stochastic Completion Time Minimization.
\end{theorem}

\subsection{Overview of Algorithm and Analysis}

We first explain the relax-and-round approach for online stochastic resource allocation problems and the existing ingredients in the literature.
At the beginning of the time horizon, before any job arrives, we solve a mathematical program relaxation of the problem to obtain fractional allocations $\vec{x}_j = (x_{ij})_{i \in M}$ of jobs $j \in J$.
Then, the stochastic environment realizes the jobs one by one by the Poisson process.
For each incoming job $j$, we use an online rounding algorithm to choose a machine $i$ based on the job's fractional allocation $\vec{x}_j$ and the rounding algorithm's internal state.
See \Cref{fig:relax-and-round} below for an illustration.

\begin{figure}[h]
\centering
\begin{tikzpicture}
	\draw[thick,fill=hkured!40,rounded corners] (-1.0,0.0) rectangle +(2.5,4.0) node[midway,align=center] {stochastic\\ environment};
	\draw[thick,fill=hkugreen!40,rounded corners] (9.0,0.0) rectangle +(2.5,4.0) node[midway,align=center] {mathematical\\ program};
	\draw[thick,fill=hkuyellow!40,rounded corners] (4.2,0.0) rectangle +(2.1,1.5) node[midway,align=center] {rounding};
	\draw[thick,->] (1.7,3.25) -- (8.8,3.25);
	\draw[thick,->] (8.8,0.75) -- (6.5,0.75);
	\draw[thick,->] (4.0,0.75) -- (1.7,0.75);
	\draw[thick,fill=white] (4.3,3.4) rectangle +(.4,.4) node[midway] {$4$};
	\draw[thick,fill=white] (4.3,4.0) rectangle +(.4,.4) node[midway] {$3$};
	\draw[thick,fill=white] (4.3,4.6) rectangle +(.4,.4) node[midway] {$2$};
	\draw[thick,fill=white] (4.3,5.2) rectangle +(.4,.4) node[midway] {$1$};
	\draw[thick,fill=white] (6.0,4.5) circle (6pt);
	\draw[thick,dotted] (4.7,3.6) -- (5.8,4.4);
	\draw[thick,dotted] (4.7,4.2) -- (5.8,4.47);
	\draw[thick,dotted] (4.7,4.8) -- (5.8,4.52);
	\draw[thick,dotted] (4.7,5.4) -- (5.8,4.6);
	\draw[thick,fill=white] (6.7,0.9) rectangle +(.4,.4) node[midway] {$4$};
	\draw[thick,fill=white] (6.7,1.5) rectangle +(.4,.4) node[midway] {$3$};
	\draw[thick,fill=white] (6.7,2.1) rectangle +(.4,.4) node[midway] {$2$};
	\draw[thick,fill=white] (6.7,2.7) rectangle +(.4,.4) node[midway] {$1$};
	\draw[thick,fill=white] (8.4,2.0) circle (6pt);
	\draw[thick,dotted] (7.1,1.1) -- (8.2,1.9) node[pos=0.45,sloped,fill=white,below=-1pt] {\footnotesize $0.2$};
	\draw[thick,dotted] (7.1,1.7) -- (8.2,1.97) node[pos=0.45,sloped,fill=white,below=-5pt] {\footnotesize $0.4$};
	\draw[thick,dotted] (7.1,2.3) -- (8.2,2.02) node[pos=0.45,sloped,fill=white,above=-5pt] {\footnotesize $0.1$};
	\draw[thick,dotted] (7.1,2.9) -- (8.2,2.1) node[pos=0.45,sloped,fill=white,above=-1pt] {\footnotesize $0.3$};
	\draw[thick,fill=white] (2.7,1.0) rectangle +(.4,.4) node[midway] {$1$};
\end{tikzpicture}
\caption{Illustration of the relax-and-round approach. The squared boxes represent the machines, and the circles represent the jobs.}
\label{fig:relax-and-round}
\end{figure}
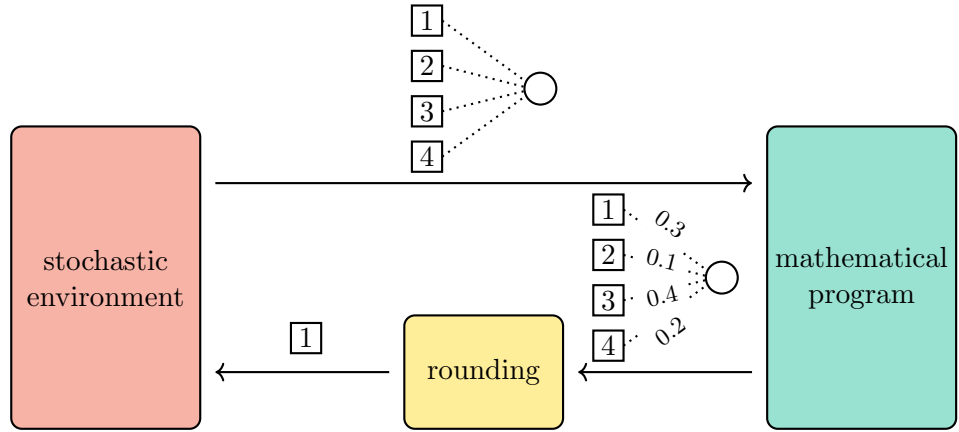

For example, one may adopt Skutella's quadratic program~\cite{Skutella:JACM:2001} into the stochastic model and apply independent rounding to allocate each incoming job $j$ to a machine $i$ with probability proportional to $x_{ij}$.
We rephrase the quadratic program below to be self-contained.
\begin{align*}
	\operatorname{minimize} \quad & 
	\int_0^\infty \sum_{i \in M} \frac{\load_i(s)^2 + \load_i(s)}{2} \,\dif{s} \\
	\operatorname{subject\ to} \quad &
	\load_i(s) \ge \sum_{j \in J : s_{ij} \ge s} x_{ij} 
	&& \forall i \in M, \forall s > 0 \\
	&
	\sum_{i \in M} x_{ij} \ge \lambda_j
	&& \forall j \in J \\[.5ex]
	&
	x_{ij}, \load_i(s) \ge 0 
	&& \forall i \in M, \forall j \in J, \forall s > 0
\end{align*}

Unfortunately, we must replace both the quadratic program and the independent rounding algorithm to get a competitive ratio better than $\frac{3}{2}$.
More precisely, the optimal objective of the quadratic program can be $\frac{3}{2}$ times larger than the expected optimal offline objective, e.g., for an instance with one machine and one unit-size job arriving at unit rate.
In a nutshell, the quadratic program's objective equals the completion time when there is exactly one job all the time, while in reality, the number of jobs follows the exponential distribution with unit rate.

The online matching literature addressed a similar issue by introducing additional linear constraints, e.g., the so-called natural Poisson constraints~\cite{HuangS:STOC:2021,HuangSY:STOC:2022,ChenHS:FOCS:2024,Yan:SODA:2024,TorricoAT:MathProg:2018}:
\[
	\forall S \subseteq J ~, \qquad \sum_{j \in S} x_{ij} \le 1 - \exp\Big(-\sum_{j \in S} \lambda_j \Big)
	~,
\]
which rely on the fact that each bin/offline vertex can be matched at most once. 
This is no longer true for scheduling problems, as each machine can take an arbitrary number of jobs. 
\Cref{sec:completion-time-lp} presents how we formulate an LP that adopts the idea of Poisson constraints into the context of scheduling problems.

Regarding the rounding algorithm, we use a reduction of general rounding to two-way rounding by~\citet{ChenHS:FOCS:2024}.
The idea is to sub-sample one or two machines with integral or half-integral allocation to them, such that the expected amount of the job allocated to each machine agrees with the fractional allocation given by the LP solution.
If two machines were sub-sampled, we apply the Mark-and-Swap algorithm as a two-way online rounding to select one of them.
Otherwise, we trivially allocate the job to the only sub-sampled machine.
\Cref{sec:completion-time-reduction} formally defines this reduction.
See \Cref{fig:decomposition} for an illustration.

\begin{figure}[h]
\centering
\begin{tikzpicture}
	\draw[thick,fill=gray!10,rounded corners] (0.0,-2.0) rectangle +(2.1,1.5) node[midway,align=center] {(trivial)\\ one-way\\ rounding};	
	\draw[thick,fill=hkuyellow!40,rounded corners] (0.0,0.0) rectangle +(2.1,1.5) node[midway,align=center] {two-way\\ rounding};
	\draw[thick,fill=hkublue!40,rounded corners] (5.0,0.0) rectangle +(2.1,1.5) node[midway,align=center] {decom-\\position};
	\draw[thick,->] (-0.2,0.75) -- (-2.7,0.75);
	\draw[thick,->] (4.8,0.75) -- (2.3,0.75);
	\draw[thick,->] (9.8,0.75) -- (7.3,0.75);
	\draw[thick,->] (4.8,0.65) -- (2.3,-1.25);	
	\draw[thick,->] (-0.2,-1.25) -- (-2.7,-1.25);		
	\draw[thick,fill=white] (7.7,0.9) rectangle +(.4,.4) node[midway] {$4$};
	\draw[thick,fill=white] (7.7,1.5) rectangle +(.4,.4) node[midway] {$3$};
	\draw[thick,fill=white] (7.7,2.1) rectangle +(.4,.4) node[midway] {$2$};
	\draw[thick,fill=white] (7.7,2.7) rectangle +(.4,.4) node[midway] {$1$};
	\draw[thick,fill=white] (9.4,2.0) circle (6pt);
	\draw[thick,dotted] (8.1,1.1) -- (9.2,1.9) node[pos=0.45,sloped,fill=white,below=-1pt] {\footnotesize $0.2$};
	\draw[thick,dotted] (8.1,1.7) -- (9.2,1.97) node[pos=0.45,sloped,fill=white,below=-5pt] {\footnotesize $0.4$};
	\draw[thick,dotted] (8.1,2.3) -- (9.2,2.02) node[pos=0.45,sloped,fill=white,above=-5pt] {\footnotesize $0.1$};
	\draw[thick,dotted] (8.1,2.9) -- (9.2,2.1) node[pos=0.45,sloped,fill=white,above=-1pt] {\footnotesize $0.3$};
	\draw[thick,fill=white] (2.7,1.0) rectangle +(.4,.4) node[midway] {$3$};
	\draw[thick,fill=white] (2.7,2.2) rectangle +(.4,.4) node[midway] {$1$};
	\draw[thick,fill=white] (4.4,1.8) circle (6pt);
	\draw[thick,dotted] (3.1,1.2) -- (4.2,1.77) node[pos=0.45,sloped,fill=white] {\footnotesize $0.5$};
	\draw[thick,dotted] (3.1,2.4) -- (4.2,1.83) node[pos=0.45,sloped,fill=white] {\footnotesize $0.5$};
	\draw[thick,fill=white] (-1.65,1.0) rectangle +(.4,.4) node[midway] {$1$};
	\draw[thick,fill=white] (2.7,-1.6) rectangle +(.4,.4) node[midway] {$1$};
	\draw[thick,fill=white] (4.4,-1.4) circle (6pt);
	\draw[thick,dotted] (3.1,-1.4) -- (4.2,-1.4) node[pos=0.5,sloped,fill=white] {\footnotesize $1.0$};
	%
	\draw[thick,fill=white] (-1.65,-1.0) rectangle +(.4,.4) node[midway] {$1$};
\end{tikzpicture}
\caption{Illustration of the reduction from general online rounding to two-way rounding.}
\label{fig:decomposition}
\end{figure}
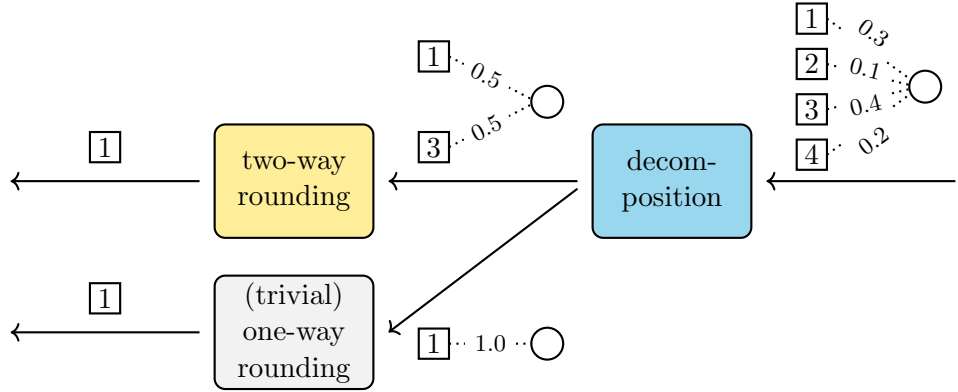



\subsection{Linear Program Relaxation}
\label{sec:completion-time-lp}

This subsection introduces our LP relaxation for the offline optimal solution, i.e., when the jobs are realized by the same Poisson process, but an (offline) algorithm can observe the realization of all jobs and then find the optimal allocation that minimizes the total completion time.
Its solution will serve as a reference for online decision-making.
Further, its optimal objective value will serve as a lower bound of the expected minimum total completion time.
Our competitive analysis will compare the online algorithm's objective to the optimal LP objective value.

Let $x_{ij}$ be the expected number of jobs of type $j$ that are allocated to machine $i$.
Further, let $y_{ijk}(s)$ be the probability that, among jobs of sizes at least $s$ allocated to machine $i$, the $k$-th largest job is of type $j$.
Finally, let $z_{ik}(s)$ denote the probability that machine $i$ is assigned with at least $k$ jobs of sizes at least $s$.

For any subset of jobs $S \subseteq J$, we write $\lambda_S \defeq \sum_{j \in S} \lambda_j$ for brevity.
Consider the following LP:
\begin{align*}
	\mathrm{minimize} \qquad & \sum_{i \in M} \int_0^\infty \sum_{k=1}^\infty k \cdot z_{ik}(s)\,\dif{s} \\
	\mathrm{subject\ to} \qquad & \sum_{i \in M} x_{ij} = \lambda_j && \forall j \in J \\
	& \sum_{k=1}^\infty y_{ijk}(s) = x_{ij} && \forall i \in M, \forall j \in J, \forall s \in (0, s_{ij}] \\
	& \sum_{j \in S} \sum_{k = 1}^\ell y_{ijk}(s) \le \E_{n \sim \Poisson(\lambda_S)} \min \{ n, \ell \} && \forall i \in M, \forall S \subseteq J, \forall \ell \ge 1, \forall s > 0\\[1ex]
	& z_{ik}(s) = \sum_{j \in J} y_{ijk}(s) && \forall i \in M, \forall k \ge 1, \forall s > 0
\end{align*}

\begin{lemma}
    \label{lem:completion-time-lp}
	The optimal LP objective is no greater than the expected minimum total completion time.
\end{lemma}

\begin{lemma}
    \label{lem:lp-poly-solvable}
	The LP can be solved in polynomial time.	
\end{lemma}

The proofs are relatively standard, and thus deferred to \Cref{app:completion-time-lp,app:lp-poly-solvable}.

\subsection{Reduction to Generalized Balls into Bins}
\label{sec:completion-time-reduction}

We treat the machines as bins and map each job to a ball as follows. 
Ideally, we want to sample two candidate machines $i \ne i'$ for every arrived job, based on its type $j$ and the LP solution, but independent of the realization and allocation of previous jobs.
The latter restriction ensures that the arrivals of balls follow a Poisson process from the viewpoint of the downstream Generalized Balls into Bins algorithm (\Cref{lem:job-ball-poisson}).
If a job is of type $j$ and we sample machines $i$ and $i'$, we say it is a ball of type $j \sim \{i, i'\}$.
We can then allocate it using the Mark-and-Swap algorithm.

However, this process cannot allocate more than half of a job $j$ to any machine $i$, even if we always sample $i$ as a candidate.
This is not a problem if $x_{ij} \le \frac{\lambda_j}{2}$ for all machines $i$ (\Cref{lem:job-ball-pairing-exist}).
If $x_{ij} > 1/2$ for some machine $i$, we need to sometimes sample only one machine $i$, in which case we call it a one-choice ball $j \sim i$, and allocate the job directly to machine $i$ without calling Mark-and-Swap.
We make the rate of sampling $j \sim i$ as small as possible, subject to preserving the marginal probability of allocating job $j$ to different machines (\Cref{lem:job-ball-allocation-preservation}).


%
%


\begin{algorithm}{Reduction: Mapping Jobs to Balls}
	\begin{itemize}
		%
		%
		\item For every job $j \in J$ such that $x_{ij} \le \frac{\lambda_j}{2}$ for all machines $i$:
		\begin{itemize} 
			\item Find arrival rates $\lambda_{j \to \{i,i'\}} \ge 0$ such that for any machine $i$:
			\begin{equation}
				\label{eqn:type-decomposition-rates}
				\frac{1}{2} \sum_{i' \ne i} \lambda_{j \to \{i, i'\}} = x_{ij}
				~.
			\end{equation}
			\item Let each arrived job $j$ be a ball $j \to \{i, i'\}$ incident to $i, i'$ with probability $\frac{\lambda_{j \to \{i, i'\}}}{\lambda_j}$.
		\end{itemize}
		\item For every job $j \in J$ such that $x_{ij} > \frac{\lambda_j}{2}$ for some machine $i$:
		\begin{itemize} 
			\item Let $\lambda_{j \to i} = 2x_{ij} - \lambda_j$ and $\lambda_{j \to \{i, i'\}} = 2 x_{i'j}$ for any $i' \ne i$.
			\item Let each arrived job $j$ be a ball $j \to i$ incident only to $i$ with probability $\frac{\lambda_{j \to i}}{\lambda_j}$, and a ball $j \to \{i, i'\}$ incident to $i, i'$ with probability $\frac{\lambda_{j \to \{i, i'\}}}{\lambda_j}$.
		\end{itemize}
	\end{itemize}
\end{algorithm}

\begin{algorithm}{Reduction: Online Allocation}
	\begin{itemize}
		\item Allocate each one-choice ball $j \to i$ to bin/machine $i$.
		\item Allocate two-choice balls using the Mark-and-Swap algorithm, letting $\succ_i$ therein be the descending order of job sizes $s_{ij}$.
	\end{itemize}
\end{algorithm}

%
%
%
%

We first argue that the mapping of jobs to balls is well defined. 

\begin{lemma}
	\label{lem:job-ball-pairing-exist}
	There are arrival rates satisfying \Cref{eqn:type-decomposition-rates}.
\end{lemma}

\begin{proof}
	By Farkas' lemma, it suffices to verify that for any $y = (y_i)_{i \in M}$ satisfying:
	\[
		\forall i \ne i' \in M ~,\quad \frac{1}{2} y_i + \frac{1}{2} y_{i'} \ge 0
		~,
	\]
	we have:
	\[
		\sum_{i \in M} x_{ij} y_i \ge 0
		~.
	\]
	
	This is trivially true if $y_i$ is non-negative for all machines $i \in M$.
	Otherwise, suppose $y_i < 0$.
	Then, we have $y_{i'} \ge - y_i$ for all $i' \ne i$.
	Hence, the left-hand-side of the inequality is at least:
	\[
		|y_i| \cdot \Big( \sum_{i' \ne i} x_{i'j} - x_{ij} \Big) ~=~ |y_i| \cdot \Big( \lambda_j - 2 x_{ij} \Big) ~\ge~ 0
		~.
	\]
\end{proof}

For notational convenience, we define $\lambda_{j \to i}$ and $\lambda_{j \to \{i, i'\}}$ to be zeros if the reduction never maps any job to such balls.

\begin{lemma}
	\label{lem:job-ball-poisson}
	For any job $j$ and any machines $i, i'$, balls $j \to i$  and $j \to \{i, i'\}$ arrive by Poisson processes with rates $\lambda_{j \to i}$ and $\lambda_{j \to \{i, i'\}}$.	
\end{lemma}

\begin{proof}
	This follows because the mapping from $j$ to balls $j \to i$ and $j \to \{i, i'\}$ is a sub-sampling independent of the previous job arrivals and their allocations.
\end{proof}

\begin{lemma}
	\label{lem:job-ball-allocation-preservation}
	For any job $j$ and any machine $i$, we have:
	\[
		\lambda_{j \to i} + \frac{1}{2} \sum_{i' \ne i} \lambda_{j \to \{i, i'\}} = x_{ij}
		~.
	\]	
\end{lemma}

\begin{proof}
	The first case follows by summing \Cref{eqn:type-decomposition-rates} over all machines $i \in M$.
	In the second case, the probabilities of mapping job $j$ to different balls sum to $1$ by definition.
	The lemma then follows by multiplying both sides with $\lambda_j$.
\end{proof}


%
%

\subsection{Proof of \Cref{thm:completion-time}}

For any machine/bin $i$ and any size-level $s > 0$, let $\load_i(s)$ be the number of jobs/balls of sizes at least $s$ that are allocated to bin $i$. 
Recall that we decompose the LP objective by size-levels:
\[
	\sum_{i \in M} \int_0^\infty \sum_{k=1}^\infty k \cdot z_{ik}(s)\,\dif{s}
	~.
\]

Further, recall that we can also write the algorithm's objective by size levels.
\[
	\sum_{i \in M} \int_0^\infty \frac{\load_i(s)^2 + \load_i(s)}{2} \,\dif{s}\tag{by \Cref{eqn:completion-time-by-size}}
	~.
\]

To prove a competitive ratio of $\Gamma = 1.435$, it suffices to consider any bin $i$ and any size-level $s > 0$ and prove that:
\[
	\frac{\load_i(s)^2 + \load_i(s)}{2} ~\le~ \Gamma \cdot  \sum_{k=1}^\infty k \cdot z_{ik}(s)
	~.
\]

Let $d_1$ be the total arrival rate of bin $i$'s incident one-choice balls with sizes at least $s$, and let $d_2$ be the normalized total arrival rate of $i$'s incident two-choice balls with sizes at least $s$:
\[
	d_1 \defeq \sum_{j : s_{ij} \ge s} \lambda_{j \to i}
	~,\quad
	d_2 \defeq \frac{1}{2} \sum_{j : s_{ij} \ge s} \sum_{i' \ne i} \lambda_{j \to \{i, i'\}}
	~.
\]

Let $d = d_1 + d_2$ be the degree of machine $i$ at size level $s$.
It equals the expected number of jobs allocated to $i$ with sizes at least $s$ because Mark-and-Swap is mean-preserving.
We suppress the dependence on machine $i$ and size-level $s$ as they are fixed in the rest of the argument.

By the Poisson arrivals of one-choice balls and the worst-case load distribution of two-choice balls (\Cref{thm:mark-swap-worst-case}):
\[
	\load_i(s) \,\succeq\, \Poisson(d_1) + \MarkSwap(d_2)
	~.
\]

By the expectations and variances of the Poisson and MarkSwap distributions (\Cref{lem:markswap-expectation,lem:markswap-variance}), the second moment of $\load_i(s)$ is at most:
\[
	\E \load_i(s)^2 \,\le\, d^2 + d - \frac{2}{e^2} \big( \lfloor d_2 \rfloor + \{ d_2 \} ^ 2 \big)
  	~,
\]
where we replace $d_1$ by $d - d_2$ after applying \Cref{lem:markswap-expectation,lem:markswap-variance}.

Hence, the expectation of $\frac{\load_i(s)^2 + \load_i(s)}{2}$ is at most:
\begin{equation}
    \label{eqn:completion-time-alg}
    \frac{1}{2} \left( d^2 + 2d - \frac{2}{e^2} \big( \lfloor d_2 \rfloor + \{ d_2 \} ^ 2 \big) \right)
\end{equation}

Next, we consider the counterpart in the LP objective:
\begin{equation}
    \label{eqn:completion-time-lp}
	\sum_{k=1}^\infty k \cdot z_{ik}(s)
	~,
\end{equation}
where:
\[
	\sum_{k=1}^\infty z_{ik}(s) = d
	~,\quad
	0 \le z_{ik}(s) \le 1
	~.
\]

However, these two conditions are insufficient.
For example, suppose that (a) bin $i$ had degree $d = 1$, but (b) none of it came from two-choice balls, i.e., $\lfloor d_2 \rfloor = 0$, and (c) exactly one such ball was allocated to bin $i$ in all cases, i.e., $z_{i1}(s) = 1$ and $z_{ik}(s) = 0$ for $k > 1$.
Then, we would have $\eqref{eqn:completion-time-alg} = \frac{3}{2}$ and $\eqref{eqn:completion-time-lp} = 1$.
Intuitively, the rest of the argument shows that (a), (b), and (c) cannot hold simultaneously. 
Unfortunately, this final part of the proof is tedious.

For each job $j$, we charge:
\[
	a_j \defeq \min \big\{ \lambda_j - x_{ij}, y_{ij1}(s) \big\}
\]
amount of the allocation $y_{ij1}(s)$ to two-choice balls, i.e., those of types $j \sim \{i, i'\}$ for some $i' \ne i$.

To make this intuition behind $a_j$ precise, let us consider two cases depending on the comparison between $x_{ij}$ and $\frac{\lambda_j}{2}$.
If less than half of job $j$ was allocated to machine $i$, i.e., $x_{ij} \le \frac{\lambda_j}{2}$, then all balls involving job $j$ and machine $i$ are two-choice balls according to the mapping from jobs to balls in \Cref{sec:completion-time-reduction}.
Hence, balls contributing to $y_{ij1}(s)$, which involve both $i$ and $j$ by definition, must be two-choice balls;
we expect $a_j = y_{ij1}(s)$.
Indeed, this holds because:
\[
	\lambda_j - x_{ij} \ge x_{ij} \ge y_{ij1}(s)
\]
where the second inequality follows by the second set of LP constraints.

If $x_{ij} > \frac{\lambda_j}{2}$, then $\lambda_j - x_{ij}$ is the total arrival rate of two choice balls $j \to \{i, i'\}$ involving job $j$ and machine $i$, according to the mapping from jobs to balls in \Cref{sec:completion-time-reduction}.
In this case, the definition of $a_j$ means that we charge $y_{ij1}(s)$ to two-choice balls as much as possible.


Further, charge the remaining:
\[
	b_j 
	\defeq y_{ij1}(s) - a_j
	\,=\, 
    \big(y_{ij1}(s) - a_j\big)^+
 	\,\le\, 
    \big( x_{ij} - y_{ij1}(s) \big) + \big( 2 y_{ij1}(s) - \lambda_j \big)^+
\]
amount of the allocation $y_{ij1}(s)$ to one-choice balls $j \sim i$.
Finally, let:
\[
	c_j \defeq x_{ij} - y_{ij1}(s)
\]
be the expected number of jobs $j$ that are allocated to machine $i$ but are \emph{not} the smallest among jobs allocated to $i$ with sizes at least $s$.

Define the sum of these three quantities over all jobs as:
\[
	a \,\defeq\, \sum_{j \in J} a_j
	~,\quad
	b \,\defeq\, \sum_{j \in J} b_j
	~,\quad
	c \,\defeq\, \sum_{j \in J} c_j
	~.
\]

By the definition of $a$ and $b$, we have:
\begin{equation}
	\label{eqn:degree-abc}
	d = a + b + c
	~,
	\quad
	d_2 \ge a
	~,
\end{equation}
and also:
\begin{equation}
	\label{eqn:z-ab}
	z_{i1}(s) = a + b \le 1
	~.
\end{equation}

Next, we derive another inequality that relates $b$ and $c$.
Recall the third set of LP constraints for $\ell = 1$, which we rephrease below:
\begin{equation}
	\label{eqn:poisson-constraint}
	\forall S \subseteq J ~,\quad \sum_{j \in S} y_{ij1}(s) \le 1 - e^{-\lambda_S}
	~.
\end{equation}

\begin{lemma}[Converse Jensen Inequality \cite{HuangS:STOC:2021}, applied to $(2x-1)^+$]
	\label{lem:converse-jensen}
	For any non-negative $y_{ij1}(s)$, $j \in J$, satisfying \Cref{eqn:poisson-constraint}, we have:
	\[
		\sum_{j \in J} (2y_{ij1}(s) - \lambda_j)^+ \le 1 - \ln 2
		~.
	\]
\end{lemma}

By \Cref{lem:converse-jensen} and the definitions of $b$ and $c$, we have:
\begin{equation}
	\label{eqn:bc}
	b \le c + 1 - \ln 2
	~.
\end{equation}

Rewriting the upper bound of the expectation of $\frac{\load_i(s)^2 + \load_i(s)}{2}$ in \Cref{eqn:completion-time-alg} in terms of $a$, $b$, and $c$, i.e., combining \Cref{eqn:completion-time-alg,eqn:degree-abc}, we have:
\[
	\E \frac{\load_i(s)^2 + \load_i(s)}{2} \,\le\, \frac{1}{2} \Big( (a + b + c)^2 + 2(a + b + c) - \frac{2}{e^2} a^2 \Big)
	~.
\]

Similarly, by \Cref{eqn:completion-time-lp,eqn:degree-abc,eqn:z-ab}, the counterpart in the LP objective is at least:
\[
	a + b + \sum_{k=2}^{\lfloor c \rfloor + 1} k + \{ c \} \cdot \big( \lfloor c \rfloor + 2 \big)
	~=~
	a + b + c + \frac{\lfloor c \rfloor^2 + \lfloor c \rfloor}{2} + \{c\} (\lfloor c \rfloor + 1)
	~.
\]

For the stated competitive ratio $\Gamma = 1.435$, we want to show that:
\begin{equation}
	\label{eqn:completion-time-competitive}
	\frac{1}{2} \Big( (a + b + c)^2 + 2(a + b + c) - \frac{2}{e^2} a^2 \Big) 
	~\le~ 
	\Gamma \cdot \Big( a + b + c + \frac{\lfloor c \rfloor^2 + \lfloor c \rfloor}{2} + \{c\} (\lfloor c \rfloor + 1) \Big)
	~,
\end{equation}
for any non-negative $a$, $b$, and $c$ that satisfy linear constraints $a + b \le 1$ by \Cref{eqn:z-ab} and $b \le c + 1 - \ln 2$ by \Cref{eqn:bc}.

Consider the derivatives of $c$ on the two sides of \Cref{eqn:completion-time-competitive}.
The derivative on the left is $a + b + c + 1 \le c + 2 < \lfloor c \rfloor + 3$, while the counterpart on the right-hand-side is $\Gamma \big( \lfloor c \rfloor + 2 \big)$.
The latter is strictly larger if $\lfloor c \rfloor \ge 1$.
Hence, we may focus on $c < 1$, where the inequality becomes:
\[
	\frac{1}{2} \Big( (a + b + c)^2 + 2(a + b + c) - \frac{2}{e^2} a^2 \Big) 
	-
	\Gamma \cdot \big( a + b + 2c \big)
	~\le~
	0
	~.
\]

The left-hand-side is quadratic in $a$ with a positive quadratic coefficient.
Hence, subject to $a + b \le 1$, it suffices to verify the inequality for $a = 0$ and $a = 1-b$.

\bigskip
\noindent{\emph{Case 1: $a = 0$.~}}
We need to show that:
\[
	\frac{1}{2} \big( (b+c)^2 + 2 (b+c) \big) - \Gamma \cdot \big( b + 2c \big) ~\le~ 0
	~,
\]
subject to $0 \le b, c \le 1$ and $b \le c + 1 - \ln 2$.

Since the left-hand-side is convex in $b$ and $c$, it suffices to prove the inequality at the vertices of a polytope defined by these linear constraints.
The vertices are $(b,c) = (0,0)$, $(1-\ln2, 0)$, $(1, \ln2)$, $(0, 1)$, and $(1,1)$.
The inequality holds trivially when $b = c = 0$.
Further, we may omit the vertices with $c = 1$ by the above discussion on the derivatives of $c$.
If $(b, c) = (1-\ln2, 0)$, we have:
\[
	\Big( \frac{3-\ln2}{2} - \Gamma \Big) (1-\ln 2) < - 0.08
	~.
\]

If $(b, c) = (1, \ln2)$, we have:
\[
	\frac{1}{2}(1+\ln2)(3+\ln2) - \Gamma (1+2\ln2) < - 0.29
	~.
\]


\bigskip
\noindent{\emph{Case 2: $a =1-b$.~}}
We need to show that:
\[
	\frac{1}{2} \Big( (1 + c)^2 + 2(1 + c) - \frac{2}{e^2} (1-b)^2 \Big) 
	-
	\Gamma \cdot \big( 1 + 2c \big)
	~\le~
	0
	~.
\]

For any $c$, the left-hand-side is maximized when $b = 1 - (\ln 2 - c)^+$.
If $\ln 2 \le c < 1$, we have $b = 1$ and the inequality becomes:
\[
	\frac{1}{2} \Big( (1 + c)^2 + 2(1+c) \Big) - \Gamma \cdot \big( 1 + 2c \big) ~\le~ 0
	~.
\]

This is convex in $c$ so it suffices to consider $c = \ln 2$ and $c = 1$.
In the former case, it is:
\[
	\frac{1}{2} \Big( (1 + \ln2)^2 + 2(1+\ln2) \Big) - \Gamma \cdot \big(1 + 2\ln2 \big) < -0.29
	~.
\]

In the latter case, it is $4 - 3\Gamma < 0$.

If $0 \le c < \ln 2$, we have $b = 1 - \ln 2 + c$ and the inequality becomes:
\[
	\frac{1}{2} \Big( (1 + c)^2 + 2(1+c) - \frac{2}{e^2} (\ln 2 - c)^2 \Big) - \Gamma \cdot \big( 1 + 2c \big) ~\le~ 0
	~.
\]

This is also convex in $c$ so it suffices to consider $c = 0$; the case of $c = \ln 2$ coincides with the corresponding case above that we have already verified.
When $c = 0$, it is:
\[
	\frac{1}{2} \Big( 3 - \frac{2}{e^2}(\ln2)^2 \Big) - \Gamma < - 10^{-5}
	~.
\]

This last case is the binding one that defines $\Gamma = 1.435$.

\section{Discussion}
\label{sec:discussion}
\subsection{General Framework}
\label{sec:general-framework}

In this section, we present a general framework that extends the analysis of two-choice balls discussed in \Cref{sec:2-choice-framework} to both multi-choice and mixed‑choice settings.

\paragraph{Separable Potential Function.}
We once again design a potential for each ball $e$, denoted as $\Phi_e : [0, 1] \to \R_{\ge 0}$.
Then, we define the overall potential by summing over all balls:
\[
    \Phi(t) \,\defeq\, \sum_{e \in E} \lambda_e \cdot \Phi_e(t)
	~.
\]

Recall that $\load_v$ denotes the load of bin $v$, i.e., the number of balls in it.
We extend the current loads of each ball $e$'s incident bins from the unordered pair (for two-choice) to the unordered tuple $S_e=(L_v)_{v\sim e}$.
The potential of ball $e$ at time $t$ depends on its load tuple $S_e$. 
That is, for functions $\Phi^{S} : [0, 1] \to \R_{\ge 0}$ to be determined, we let:
\[
	\Phi_e(t) = \Phi^{S_e}(t)
	~.
\]
For ease of notation, we further denote $\mathbf{0}$ to be the zero tuple that only contains element(s) equal to $0$.

\paragraph{Local Algorithms.}
We consider algorithms that allocate each incoming ball $e$ based on the local information of its current load tuple $S_e$ (for instance, Greedy). 
Let $x^{S_e}_v$ denote the probability that a local algorithm allocates an arrived ball $e$ with load tuple $S_e$ to bin $v$.

\paragraph{Degree Vectors.}
Let $\loadeset_v$ be the set of the load tuples of $v$'s incident balls.
For any $S \in \loadeset_v$, let $\Lambda^{S}_v$ be the normalized total arrival rate of its incident balls $e \sim v$ with load tuple $S_e = S$, i.e.:
\[
	\Lambda^{S}_v ~\defeq \frac{1}{|S|} \sum_{e \sim v \,:\, S_e = S} \lambda_e
	~,
\]
where $|S|$ denotes the number of elements in tuple $S$.
Observe that $|S|$ equals $|e|$, i.e., the size of a hyper-edge/ball $e$  with load tuple $S_e = S$.

Further, suppose the vector $\vec{\Lambda}_v = ( \Lambda^{S}_v )_{S \in \loadeset_v}$ lies in a polytope $\polytope_v$.
Let $\support(\polytope_v)$ be the set of vectors spanning the polytope $\polytope_v$.
By the definition of $\Lambda^{S}_v$ and the degree of bin $v$, we have:
\[
	\sum_{S\in \loadeset_v} \Lambda_v^{S} = d_v
	~.
\]

The framework further allows additional linear constraints.

\begin{theorem}
	\label{thm:potential-based-general-extention}
	Consider a maximization problem with function $f : \Z_{\ge 0} \to \R_{\ge 0}$, a bound $g : \R_{\ge 0} \to \R_{\ge 0}$, and a local algorithm.
	The algorithm guarantees an amortized bound $g$ if there are potential functions $\Phi^{S}$ satisfying:
	\begin{enumerate}
		\item For any bin $v \in V$:
		\hspace*{\fill} (starting condition)
        \[
            \sum_{e \sim v} \, \frac{\lambda_e}{|e|} \cdot \Phi^{S_e =\, \mathbf{0}}(0) \,\ge\, g(d_v)
            ~;
        \]
		\item For any load tuple $S$, $\Phi^{S}(1) = 0$;
		\hspace*{\fill} (terminal condition)
		\item For any bin $v \in V$ and its load $\load_v$, for any $\vec{\Lambda} \in \support( \polytope_v )$:
		\hspace*{\fill} (monotonicity at limits)
		\[
			\sum_{S \in \loadeset_v} \Lambda^{S} |S| \, x^{S}_v \Big( f(\load_v + 1) - f(\load_v) + \frac{\dif}{\dif{t}} \Phi^{S}(t) \Big) 
			\,\ge\, 
			\sum_{S \in \loadeset_v} \Lambda^{S} |S| \, x^{S}_v \sum_{S \in \loadeset_v} \Lambda^{S} |S| \, \Big( \Phi^{S}(t) - \Phi^{S_+}(t) \Big)
			~,
		\]
		where $S_+$ denotes the load tuple obtained from increasing the element $\load_v$ by $1$.
		\item For any bin $v \in V$ and its load $\load_v$, there is an order of $\vec{\Lambda} \in \support(\polytope_v)$ by which both
		\[
			\sum_{S \in \loadeset_v} \Lambda^{S} \, |S| \,x^{S}_v
		\]
		and
		\[
			\sum_{S\in \loadeset_v} \Lambda^{S} \, |S| \,\Big( \Phi^{S}(t) - \Phi^{S+}(t) \Big)
		\]
		are non-decreasing.
		\hspace*{\fill} (Chebyshev's condition)
	\end{enumerate}
	The theorem also holds for minimization problems, changing the inequalities' directions in 1) and 3), and letting the two terms have opposite monotonicities in 4).
\end{theorem}
We defer the proof to \Cref{app:general-framework-extension}.


\subsection[Multi-Choice Balls and c-Matching]{Multi-Choice Balls and $c$-Matching}

This section studies the multi-choice balls setting, in which each ball is incident to $k$ bins, together with $c$-matching, i.e., the maximization problem with concave function $f(\load) = \min \{ \load, c \}$.
Suppose the bins' degrees are at most $d$.
We will modify Greedy as follows.
For any bin $v \in V$, define
\[
	r_v(t) \defeq \big(c - \load_v\big)^+
\]
to be the \emph{remaining capacity} of the bin.

\begin{algorithm}{Proportional Greedy Algorithm (for Multi-Choice Balls and $c$-Matching)}
    When a ball $e$ arrives at time $t$:
    \begin{itemize}
        \item Allocate it to an incident bin $v \sim e$ with probability proportional to $r_v(t)$.
    \end{itemize}
\end{algorithm}

\begin{corollary}
    \label{cor:k-choice}
    
    Proportional Greedy Algorithm guarantees a bound:
    \[
    	g(d_v) \,=\, \frac{\E_{j \sim \mathrm{Poisson}(dk)} \min \{ j, ck \}}{dk} \cdot d_v
    	~,
    \]
    i.e., the expected number of balls allocated to bins, each of which has capacity $c$, is at least:
    \[
	    \frac{\E_{j \sim \mathrm{Poisson}(dk)} \min \{ j, ck \}}{dk} \cdot \sum_{v \in V} d_v
	    ~.
    \]
\end{corollary}

Note that Proportional Greedy coincides with Greedy if the bins have unit capacities, (i.e., when $c = 1$).
In other words, Greedy yields optimal amortized bounds for 
(i) $k = 2$ and any $c \ge 1$,
(ii) $c = 1$ and any $k \ge 2$, and trivially
(iii) $k = 1$ and any $c \ge 1$ (since the algorithm makes no decision when each ball has only a single choice).
We do not know if the modification to the algorithm is necessary when $c \ge 2$ and $k \ge 3$.
In fact, we conjecture that the original Greedy algorithm gets the same bound, although we could not prove it using the current analysis.
Resolving the conjecture affirmatively would extend the optimal amortized bounds to arbitrary convex and concave functions via the decompositions in \Cref{lem:concave-decompose,lem:convex-decomposition}.
In its current form, however, we could not apply the decompositions because Proportional Greedy allocates the ball differently for different capacities $c$.

Consider potential functions:
\begin{equation}
	\label{eqn:k-choice-potential}
	\Phi^{S_e}(t) = \E_{j \sim \mathrm{Poisson}(dk(1-t))} \frac{\min \bigl\{ j \,,\, \sum_{v \sim e} (c - \load_v)^+ \bigr\}}{dk}
	~.
\end{equation}

Omitting the normalizing $\frac{1}{dk}$ factor, this equals the expected number of balls that could be allocated to the remaining capacities of bins incident to ball $e$, should balls $e$ arrive with rate $dk$ in the remaining time horizon from $t$ to $1$.
The proof follows by verifying the conditions of \Cref{thm:potential-based-general-extention}, which we defer to \Cref{app:k-choice}.

\subsection{Mixed Numbers of Choices}
Recall that Greedy gives the optimal bounds for the following cases:
(i) matching, i.e., $f(\load) = \min \{\load, 1\}$, (ii) when all balls have $k$-choices for any $k \ge 2$ (\Cref{cor:k-choice}), and trivially
(iii) $k = 1$.
Therefore, it is natural to consider the same algorithm for instances whose balls are incident to different numbers of bins.
The ability to handle balls with mixed numbers of choices is important for downstream applications such as Online Stochastic Matching;
see, e.g., \citet{ChenHS:FOCS:2024} for a framework for reducing Online Stochastic Matching to Stochastic Online Correlated Selection, which corresponds to Generalized Balls into Bins with $f(\load) = \min \{\load, 1\}$ and balls incident to $1$ or $2$ bins.

Optimistically, one may ask whether Greedy gets the best of many worlds, i.e., the optimal amortized bounds in \Cref{cor:c-matching} and \Cref{cor:k-choice} for all balls $e$ with their respective $k = |e|$.
That is, is it true that the expected number of non-empty bins at the end is at least
\[
	\sum_{e \in E} \frac{\E_{j \sim \mathrm{Poisson}(d|e|)} \min \{ j, |e| \}}{d|e|} \cdot \lambda_e
	~?
\]

Suppose for each ball $e$ we consider the same potential as in the proof of \Cref{cor:k-choice}, i.e., \Cref{eqn:k-choice-potential} with $c = 1$ and $k = |e|$.
Under the potential-based amortized analysis framework of \Cref{thm:potential-based-general-extention}, we can prove the above bound if each ball is incident to $1$ or $2$ bins, by verifying the four conditions of \Cref{thm:potential-based-general-extention}.

Further, the framework is general enough to allow an upper bound on the total arrival rate of $1$-choice balls incident to a bin.
For example, \citet{ChenHS:FOCS:2024} effectively reduced an Online Stochastic Matching instance to a Generalized Balls into Bins instance with $1$-choice and $2$-choice balls, such that the total rate of $1$-choice balls incident to any bin is at most $1 - \ln 2$.

\begin{corollary}
    \label{cor:Mixed-choice}
	For an instance where balls are incident to $1$ or $2$ bins, and the total arrival rate of $1$-choice balls incident to each bin is at most $1 - \ln 2$, Greedy guarantees an amortized bound of:
	\[
		1 + \frac{\ln{2}}{1-\ln{2}} \frac{1}{e^2} - \frac{1}{1-\ln{2}} \frac{1}{2e} > 0.706
		~.
	\]
	That is, the expected number of non-empty bins is at least:
	\[
		\Big( 1 + \frac{\ln{2}}{1-\ln{2}} \frac{1}{e^2} - \frac{1}{1-\ln{2}} \frac{1}{2e} \Big) \cdot
        \sum_{v \in V} d_v
		~.
	\]
\end{corollary}

The above bound $0.706$ is the same as the best competitive ratio achievable by two-choice Online Stochastic Matching algorithms~\cite{JailletL:MOR:2014,HuangSY:STOC:2022}.
We defer the proof of \Cref{cor:Mixed-choice} to the \Cref{app:mixed-choice}.
 
For an instance where balls may have $1$ to $3$ choices (or more), the first three conditions of \Cref{thm:potential-based-general-extention} still hold, but the fourth fails.
It is an interesting future research direction to study the reduction of Online Stochastic Matching to Generalized Balls into Bins with mixed numbers of choices, and design algorithms for the latter.
Progress in this direction will likely improve the state-of-the-art of Online Stochastic Matching.

\subsection{Approximate MDP Algorithm}

Generalized Balls into Bins can be viewed as a Markov decision process (MDP) in a bounded and continuous time horizon.
For example, we may consider a state space comprised of pairs $(S, t)$ of bins' loads $S = (\load_v)_{v \in V}$ and time $t$.
The action space is the set of mappings $\pi$ from each ball to an incident bin.
Given a state and an action, nature decides if a ball arrives at the current time according to the Poisson process.
If a ball $e$ arrives and is mapped to bin $v$ by the action, bin $v$'s load increases by $1$.
Otherwise, the loads stay the same.
In either case, time elapses at a unit rate.

As a result, the optimal policy for an instance, e.g., a maximization problem with function $f$, is characterized by a dynamic program.
Let $V(S, t)$ be the value-to-go of state $(S, t)$.
The optimal policy allocates each ball $e$ to an incident bin $v$ to maximize the sum of the immediate increase of the objective and the value-to-go of the new state.
Formally, the Hamilton–Jacobi–Bellman (HJB) equation of the dynamic program is:
\[
	\sum_{e \in E} \lambda_e \cdot \max_{v \sim e} \Big( f(\load_v +  1) - f(\load_v) + V \big( ( \load_v +  1, S_{-v} ), t \big) - V(S, t) \Big)
	\,+\, 	
	\frac{\partial}{\partial t} V \big(S, t\big) 	
	\,=\,
	0
	~,
\]
where $( \load_v +  1, S_{-v} )$ denotes the load tuple obtained by increasing element $\load_v$ by $1$.
Correspondingly, the optimal policy maps each ball $e$ to the maximizer of the maximization problem above.

From this perspective, we may interpret the potential functions in this section as approximate value-to-go functions that satisfy the above differential equation with inequality.
Recall that $\Phi^{S_e}(t)$ is the potential of ball $e$ at time $t$ when its load tuple is $S_e$.
The overall potential is:
\[
	\Phi(S, t) \,=\, \sum_{e \in E} \lambda_e \cdot \Phi^{S_e}(t)
	~,
\]
where we change the notation from $\Phi(t)$ to $\Phi(S, t)$ to make the dependence on $S$ explicit, and to be consistent with the notation of value-to-go functions.
Further, let $\pi$ be the mapping of balls to incident bins by Greedy.
The analysis framework (proof of \Cref{thm:potential-based-general-extention}) implies that:
\begin{align*}
	&
	\sum_{e \in E} \lambda_e \cdot \max_{v \sim e} \Big( f(\load_v +  1) - f(\load_v) + \Phi \big( ( \load_v +  1, S_{-v} ), t \big) - \Phi(S, t) \Big) 
	\,+\, 
	\frac{\partial}{\partial t} \Phi \big(S, t\big) 
	\\
	& \quad
	\,\ge\, 
	\sum_{e \in E} \lambda_e \cdot \Big( f(\load_{\pi(e)} +  1) - f(\load_{\pi(e)}) + \Phi \big( ( \load_{\pi(e)} +  1, S_{-\pi(e)} ), t \big) - \Phi(S, t) \Big) 
	\,+\, 
	\frac{\partial}{\partial t} \Phi \big(S, t\big) 
	\,\ge\,
	0
	~,
\end{align*}
where $\sum_{e \in E} \lambda_e \big( f(\load_{\pi(e)} +  1) - f(\load_{\pi(e)}) \big)$ corresponds to the increase of the algorithm's objective $A(t)$, $\sum_{e \in E} \lambda_e \big( \Phi ( ( \load_{\pi(e)} +  1, S_{-\pi(e)} ), t ) - \Phi(S, t) \big)$ is the decrease of potential $\Phi(S, t)$ due to the allocation at time $t$, and finally, $\frac{\partial}{\partial t} \Phi \big(S, t\big)$ corresponds to the change of potential $\Phi(S, t)$ due to time elapses.

Further, the bounds in this section can also be achieved by an approximate MDP algorithm that allocates each arrived ball $e$ at time $t$ to an incident bin $v \sim e$ that maximizes:
\[
	f(\load_v +  1) - f(\load_v) + V \big( ( \load_v +  1, S_{-v} ), t \big)
	~,
\]
where we omit $- \Phi(S, t)$ since it is a fixed term for all incident bins.

\bibliography{balls-into-bins}

\appendix

\section{Discussion on Asymptotic Equivalence}
\label{sec:asymptotic-equivalence}

\subsection{Sum of Convex or Concave Functions}

First, we show that the Poisson arrival model is harder in the sense that the guarantees therein directly apply to the discrete-time model.

\begin{lemma}
	Consider any convex or concave function $f$ and any algorithm $A$.
	There is an algorithm $A'$ such that the expected objective of $A'$ for $n$ balls is at least as good as the expected objective of $A$ for $n' \sim \Poisson(n)$ balls.
\end{lemma}

\begin{proof}
	We prove the claim for convex functions $f$ through a coupling argument. 
	The proof for concave functions is similar.
	
	First, consider a mapping from the ``missing balls'' $n'+1, n'+2, \dots, n$ to the ``extra balls'' $n+1, n+2, \dots, n'$ for different realization of $n' \sim \Poisson(n)$. 
	Formally, define a randomized strictly monotone mapping $\sigma : \{1, 2, \dots, n\} \to \{n+1, n+2, \dots \}$ such that for any $j > n$:
	\begin{equation}
		\label{eqn:poisson-harder-randomized-mapping}
		\E_{n' \sim \Poisson(n)}  \sum_{n' < i \le n} \mathbf{1}_{\sigma(i-n') = j}  ~=~ \Pr_{n' \sim \Poisson(n)} \big[ n' \ge j \big]
		~.
	\end{equation}
	%
	Such a mapping exists because $\E_{n' \sim \Poisson(n)} = n$.
	
	Next, we define algorithm $A'$ for a fixed number of $n$ balls. 
	We draw $n' \sim \Poisson(n)$ and sample the randomized mapping $\sigma$.
	Then, we construct a sequence of balls so that:
	\begin{itemize}
	\item The first $\min\{n', n\}$ balls are identical to the input sequence;
	\item Each of the remaining balls $n' < i \le n$ will be the $\sigma(i-n')$-th ball in the sequence; and 
	\item Sample fake balls independently for the remaining slots.
	\end{itemize}
	
	Algorithm $A'$ follows algorithm $A$'s allocation of the corresponding balls in the above sequence.
	Since $\sigma$ is a monotone mapping, this algorithm $A'$ can be implemented online by sampling the fake balls on the fly. 
	
	\bigskip
%
	The design of $A'$ naturally defines a coupling of a run of $A'$ on $n$ balls and a run of $A$ on the first $n' \sim \Poisson(n)$ balls of the sequence.
	Let $A(i)$ denote the increase of the objective when algorithm $A$ allocates the $i$'th ball of the sequence. 
	Then, the expected objective of $A$ is:
	\[
		\E_{n' \sim \Poisson(n)} \sum_{i=1}^{n'} A(i)
		~=~
		\sum_{i = 1}^\infty A(i) \cdot \Pr_{n' \sim \Poisson(n)} \big[ n' \ge i \big]
		~.
	\]
	
	By the convexity of function $f$, the expected objective of $A'$ is at most:
	\[
		\E_{n' \sim \Poisson(n)} \bigg[ \sum_{1 \le i \le \min\{n', n\}} A(i) + \sum_{n' < i \le n} A\big( \sigma(i-n') \big) \bigg]
		~.
	\]
	
	Further, by linearity of expectation and applying \Cref{eqn:poisson-harder-randomized-mapping} to the second part, this equals:
	\[
		\sum_{i = 1}^\infty A(i) \cdot \Pr_{n' \sim \Poisson(n)} \big[ n' \ge i \big]
		~.
	\]	
	
	We conclude that the expected objective of $A'$ on $n$ balls is at most the expected objective of $A$ on $n' \sim \Poisson(n)$ balls.
\end{proof}

Next, we show that the opposite direction holds up to $1 - o(1)$ for monotone concave functions.

\begin{lemma}
	Consider any non-negative, monotone, and concave function $f$ and any algorithm $A$.
	There is an algorithm $A'$ such that the expected objective of $A'$ on $n' \sim \Poisson(n)$ balls is at least a $1-o(1)$ factor of the expected objective of $A$ on $n$ balls.
\end{lemma}

\begin{proof}
	Fix any distribution of balls.
	Let $A(i)$ denote the expected increase in the objective when algorithm $A$ allocates the $i$-th ball.
	This is non-negative by the monotonicity of function $f$.
	We further normalize $f(1) - f(0)$ to be $1$ so that $A(i) \le 1$ for any $i \ge 1$.
	
	If $A(i)$ is non-increasing, then lemma follows by $\E_{n' \sim \Poisson(n)} (n-n') = o(n)$.
	Otherwise, we construct $A'$ to ``monotonizes'' $A$ as follows.
	First, define a mapping $\sigma : \Z_{\ge 1} \to \Z_{\ge 1}$ such that $A(\sigma(i))$ is approximately monotone:
	\begin{equation}
		\label{eqn:asymp-equiv-approximate-monotone}
		\forall i < i' : \quad A \big( \sigma(i) \big) ~\ge~ \Big( 1 - \frac{1}{n} \Big) A \big( \sigma(i') \big)
		~,
	\end{equation}
	and it approximately dominates $A(i)$:
	\begin{equation}
		\label{eqn:monotonization}	
		\forall i : \quad A \big( \sigma(i) \big) ~\ge~ \Big( 1 - \frac{1}{n} \Big) \cdot A(i)
		~,
	\end{equation}
	where $\frac{1}{n}$ can be replaced by any $o(1)$ function.
	This mapping can be constructed as follows:
	\begin{enumerate}
		\item For any $k \ge 0$, define:
		\[
			I_k ~\defeq~ \bigg\{~ i : \Big(1-\frac{1}{n}\Big)^{k+1} < A(i) \le \Big(1-\frac{1}{n}\Big)^k ~\bigg\}
			~.
		\]
		\item If $I_0$ is infinite, define $\sigma(i)$ to be the indices in $I_0$ in ascending order.
		\item Otherwise, define the first $|I_0|$ values of $\sigma(i)$ to be the indices in $I_0$ in ascending order, and resurvely consider $I_1$, $I_2$, and so forth.
	\end{enumerate}
	
	Let $A'$ allocate the $i$-th ball following how algorithm $A$ allocates the $\sigma(i)$-th ball, generating fake balls $\sigma(i-1) < j < \sigma(i)$ on the fly.
	By the concavity of $f$, the objective of $A'$ is at least:
	\[
		\E_{n' \sim \Poisson(n)} \sum_{i = 1}^{n'} A\big( \sigma(i) \big)
		~.
	\]
	
	By the approximate monotonicity of $A \big( \sigma(i) \big)$ (\Cref{eqn:asymp-equiv-approximate-monotone}) and $A \big( \sigma(i) \big) \ge 0$, this at least:
	\[
		\big(1 - o(1) \big) \E_{n' \sim \Poisson(n)} \Big( 1 - \frac{(n-n')^+}{n} \Big) \sum_{i = 1}^n A\big( \sigma(i) \big) 
		~.
	\]
	
	Further, by $\E_{n' \sim \Poisson(n)} (n-n') = o(n)$, it is at least:
	\[
		\big( 1 - o(1) \big) \sum_{i = 1}^n A\big( \sigma(i) \big) 
		~.
	\]
	
	Finally, this is at least $\big( 1 - o(1) \big) \sum_{i = 1}^n A(i)$ by \Cref{eqn:monotonization}.
%
\end{proof}

Similar claim does not hold for general convex functions or general non-monotone concave functions.
For example, consider two bins, a single ball between them, and function $f(\load) = 0$ for $\load \le 1$, and $f(\load) = \infty$ otherwise. 
If we have exactly two balls, the optimal objective would be zero.
If we sample the number of balls from $\Poisson(2)$, however, the objective will be infinite with a constant probability.
The counterexample for general non-monotone concave functions is similar.

\subsection{Maximum Load}

For the original maximum load $\max_v \load_v$, instead of the softmax potential function $\sum_v e^{\eta \load_v}$, we need a different argument.
The critical reason is $\max_v \load_v$ is not supermodular (over the integer lattice)---putting one more ball to bin $1$ would increase the maximum load by $1$ if the loads of two bins are $(1, 1)$, but would not affect the objective for some larger load combinations such as $(2, 3)$.

Since the graphical Balls into Bins problem focuses on asymptotic bounds on the maximum load, it suffices to show that bounds of the discrete-time and Poisson arrival models differ by at most a constant factor.

\begin{lemma}
Given any algorithm $A$ for the original graphical Balls into Bins problem where balls arrive in discrete time steps, such that $A$'s upper bound on the expected maximum load increases at most linearly in the number of balls, there is an algorithm $A'$ for Poisson arrival of balls, guaranteeing an upper bound at most twice the bound of $A$.
	The opposite direction also holds.
\end{lemma}

\begin{proof}
	Consider any fixed set of bins $V$, and any distribution over balls $\{u, v\}$. 
	If we have a fixed number of balls $n$, it becomes an instance of the graphical Balls into Bins problem.
	If the number of balls is $n' \sim \Poisson(n)$, it is an instance of the Generalized Balls into Bins problem.
	
	First, consider any given algorithm $A$ for graphical Balls into Bins.
	Let $\alpha(n)$ denote its upper bound on the expected maximum load for an instance with $n$ balls, where we hide the dependence on instance-specific parameters such as graph expansion.
	%
	In the Poisson arrival model, let $A'$ be an algorithm that allocates the first $n$ balls using algorithm $A$, and then allocates the next $n$ balls (if applicable) using a fresh copy of algorithm $A$, and so forth.
	Then, $A'$ guarantees the expected maximum load is at most:
	\begin{align*}
		&
		\sum_{i = 1}^\infty \Pr_{n' \sim \Poisson(n)} \big[ (i-1) \cdot n < n' \le i \cdot n \big] \cdot \alpha(in) \\[-0.5ex]
		& \qquad
		~\le~ \alpha(n) \cdot \sum_{i = 1}^\infty \Pr_{n' \sim \Poisson(n)} \big[ (i-1) \cdot n < n' \le i \cdot n \big] \cdot i \\
		& \qquad
		~=~ \alpha(n) \cdot \E_{n' \sim \Poisson(n)} \left\lceil \frac{n'}{n} \right\rceil \\[1ex]
		& \qquad
		~\le~ \alpha(n) \cdot \E_{n' \sim \Poisson(n)} \left( \frac{n'}{n} + 1 \right) ~=~ 2\alpha(n)
		~.
	\end{align*}
	
	Conversely, given any algorithm $A$ for the Poisson arrival model, with upper bound $\alpha(n)$ on the expected maximum load for an instance with $n' \sim \Poisson(n)$ balls.
	Let $\alpha'(n')$ denote the expected maximum load conditioned on the value of $n'$.
	We have:
	\begin{align*}
		\alpha(n) 
		&
		~\ge~ \sum_{n' \sim \Poisson(n)} \alpha'(n') \\
		&
		~\ge~ \Pr_{n' \sim \Poisson(n)} \big[ n' \ge n \big] \cdot \alpha'(n) \\[1ex]
		&
		~\ge~ \frac{1}{2} \cdot \alpha'(n)
		~.
	\end{align*}

	In other words, directly running $A' = A$ gives an upper bound of $\alpha'(n) \le 2 \alpha(n)$.
\end{proof}

\section{Deferred Proofs from \Cref{sec:greedy}}
\label{app:greedy}




\subsection{Proof of \Cref{cor:c-matching}}
\label{app:c-matching}

We construct potential functions as follows:
\[
	\Phi^{(L_v,L_u)}(t) = \E_{j \sim \mathrm{Poisson}(2d(1-t))} \frac{\min \bigl\{ j \,,\, \sum_{v \sim e} (c - \load_v)^+ \bigr\}}{2d}
	~.
\]

The starting condition at $t = 0$ holds because for any bin $v$ and its degree $d_v$:
\[
\Phi^{(0,0)}(0) = \E_{j \sim \mathrm{Poisson}(2d)}\frac{\min \{ j, 2c \}}{2d} 
~.
\]
	
	The terminal condition also holds because $\Phi^{S}(1) = 0$ for any load pair $S$ by definition.

It remains to verify the last two conditions of \Cref{thm:potential-based-general}.
Consider an arbitrary bin $v \in V$ and its load $\load_v$.
The subset of consistent load pairs are:
\[
	\loadeset_v \,=\, \Big\{ (\load_v, i) : 0 \le i \le c \Big\}
	~,
\]
where load pairs $(\load_v, i)$ for $i > c$ are omitted since the argument can treat them as $(\load_v, c)$---the other bin has zero remaining capacity for these pairs.

For any bin $v \in V$ and its load $L_v$, and any load pair $S = (L_v, i)$, the monotonicity inequality we need is that:
\[
	\mathbf{1}_{\load_v < c} + \frac{\dif{}}{\dif{t}} \Phi^{(\load_v,i)}(t) 
	\,\ge\,
	 2d_v \, \Big( \Phi^{(\load_v,i)}(t) - \Phi^{(\load_v+1,i)}(t) \Big)
	~.
\]
Since the right-hand-side is non-negative and $d_v \le d$, we only need to verify the case of $d_v = d$, where it holds with equality, i.e.:
\begin{equation}
	\label{eqn:c-matching-monotone}
	1 + \frac{\dif{}}{\dif{t}} \Phi^{(\load_v,i)}(t)
	\,=\,
	2d \, \Big( \Phi^{(\load_v,i)}(t) - \Phi^{(\load_v+1,i)}(t) \Big)
	~.
\end{equation}
	
This follows by two simple facts, which we phrase as lemmas below as they will be useful in other examples too.

\begin{lemma}
\label{lem:cap-expectation-change-cap}
    For any integer $k \ge 1$ and any distribution $D$ over $\Z_{\ge 0}$: 
    \[
	    \E_{j \sim D} \min \{ j, k \} -  \E_{j \sim D} \min \{ j, k-1 \}  
    	\,=\,
    	\Pr_{j \sim D} \big[\, j \ge k \,\big]
    	~.
    \]    
\end{lemma}

\begin{proof}
	It follows because $\min \{ j, k \} - \min \{ j, k-1 \}$ equals $1$ when $j \ge k$, and $0$ otherwise.
\end{proof}

\begin{lemma}
    \label{lem:cap-expectation-change-rate}
    For any integer $k \ge 1$: 
    \[
    	\frac{\dif{}}{\dif{\lambda}} \E_{j \sim \mathrm{Poisson}(\lambda)} \min \big\{ j \,,\, k \big\}
    	\,=\, 
    	\Pr_{j \sim \mathrm{Poisson}(\lambda)} \big[\, j < k \,\big]
	    ~.
    \]    
\end{lemma}

\begin{proof}
	Consider increasing the arrival rate from $\lambda$ to $\lambda + \delta$ for a small $\delta$.
	The probability of having two or more extra arrivals from the increased rate $\delta$ is $O(\delta^2)$.
	If that did not happen, then the increased rate makes a difference only if (1) there were strictly less than $k$ arrivals from rate $\lambda$, and (2) there was an extra arrival from rate $\delta$.
	These two events happen independently with probabilities $\Pr_{j \sim \mathrm{Poisson}(\lambda)} \big[\, j < k \,\big]$ and $\delta - O(\delta^2)$.
	Overall, the expectation increases by:
	\[
		\Pr_{j \sim \mathrm{Poisson}(\lambda)} \big[\, j < k \,\big] \cdot \delta + O(\delta^2)
		~.
	\]

	Dividing this by $\delta$ and letting $\delta$ tend to zero proves the lemma.
\end{proof}

By the definition of $\Phi^{(\load_v, i)}$ and as a corollary of \Cref{lem:cap-expectation-change-rate} with $k = (c - \load_v)^+ + (c - i)^+$ and $\lambda = 2d(1-t)$, the left-hand-side of \Cref{eqn:c-matching-monotone} equals:
\[
	1 - \Pr_{j \sim \mathrm{Poisson}(\lambda)} \big[\, j < k \,\big] = \Pr_{j \sim \mathrm{Poisson}(\lambda)} \big[\, j \ge k \,\big]
	~,
\]
which equals the right-hand-side of \Cref{eqn:c-matching-monotone} as a corollary of \Cref{lem:cap-expectation-change-cap} for the same $k = (c - \load_v)^+ + (c - i)^+$ and distribution $D = \Poisson\big(2d(1-t)\big)$.

\subsection{Proof of \Cref{lem:concave-decompose}}
\label{app:concave-decompose}
	The coefficient of $f(c)$ is:
	\[
		2 \cdot \min\big\{ \load, c \big\} - \min\big\{ \load, c+1 \big\} - \min\big\{ \load, c-1 \big\}
		~,
	\]
	which equals $1$ for $c = \load$ and $0$ for other values of $c$. 
%
%

\subsection{Proof of \Cref{cor:concave-amortized}}
\label{app:concave-amortized}

By \Cref{lem:concave-decompose}, we have:
\begin{align*}
	\E \sum_{v \in V} f(\load_v)
	&
	~=~ 
	\E \sum_{v \in V} \sum_{c=1}^\infty \Big( 2 f(c) - f(c-1) - f(c+1) \Big) \cdot \min \big\{\load_v , c\big\}
	\\
	&
	~=~
	\sum_{c=1}^\infty \Big( 2f(c) - f(c-1) - f(c+1) \Big) \E \sum_{v \in V} \min \big\{\load_v , c\big\}~.
\end{align*}

Substituting \Cref{cor:c-matching} and grouping terms by $d_v$, we get:
\begin{align*}
	\E \sum_{v \in V} f(\load_v)
	&
	~\ge~ \sum_{v \in V} d_v \sum_{c=1}^\infty \Big(   2f(c) - f(c-1) - f(c+1) \Big) \cdot \frac{\E_{j \sim \mathrm{Poisson}(2d)} \min \{ j, 2c \}}{2d} \\
	&
    ~=~
	\sum_{v \in V}d_v  \sum_{c=0}^\infty
    \Big( 2 f(c) - f(c-1) - f(c+1) \Big) \cdot
    \frac{\E_{j \sim \mathrm{Poisson}(2d)} \min\{ \lfloor \nicefrac{j}{2} \rfloor,c\} + \min\{ \lceil \nicefrac{j}{2} \rceil , c \}}{2d}
	\\
	&
	~=~ \sum_{v \in V} d_v \cdot \E_{j \sim \mathrm{Poisson}(2d)} \frac{f( \lfloor \nicefrac{j}{2} \rfloor ) + f( \lceil \nicefrac{j}{2} \rceil )}{2d}
	~.
\end{align*}
Here, the last step follows by applying the decomposition in \Cref{lem:concave-decompose} in the opposite direction.

\bigskip
\noindent
\textbf{Remark.}
{Alternatively, we can also prove \Cref{cor:concave-amortized} by directly constructing the potential functions and applying \Cref{thm:potential-based-general}.

Define $\Delta f(j \,\vert\, S_e)$ to be the maximum total increase of function values, when we allocate $j$ balls to the bins incident to $e$ subject to load pair $S_e$.
Formally,
\[
	\Delta f(j \,\vert\, S_e) 
	\,\defeq\,
 	\max_{x \,:\, \sum_{v \sim e} x_v = j} ~ \sum_{v \sim e} \big( f(\load_v + x_v) - f(\load_v) \big)
 	~.
 \]

We define the potential of a ball $e$ with load pair $S_e$ as:
 \[
	\Phi^{S_e}(t) \,=\, \frac{1}{2d} \E_{j \sim \mathrm{Poisson}(2d(1-t))} \Delta f(j \,\vert\,S_e)
	~,
\]
which is the expected increase of the function values per unit of arrival rate, should balls $e$ arrive in the remaining time horizon from $t$ to $1$ at the maximum rate $2d$ allowed by the degree bound $d$, and if they were allocated optimally to the bins incident to $e$ subject to the existing loads.

The corollary then follows by verifying the conditions of \Cref{thm:potential-based-general} with the above potentials.
}

\subsection{Proof of \Cref{lem:convex-decomposition}}
	Define $\Delta f(\load) \defeq f(\load) - f(\load - 1)$ for $\load \in \Z_{\ge 0}$, where we artificially let $f(-1) = 0$ for notational convenience.
 	Further, define $\Delta^2 f(\load) \defeq \Delta f(\load) - \Delta f(\load - 1)$ for positive integers $\load$.
	We have that:
	\begin{align*}
		f(\load) 
		&
		\,=\,
		\sum_{i=1}^\load \Delta f(i) \\
		&
		\,=\,
		\sum_{i=1}^\load \sum_{j=1}^i \Delta^2 f(j) \\
		&
		\,=\,
		\sum_{j=1}^\load \Delta^2 f(j) \cdot (\load - j + 1) \\
		&
		\,=\,
		\sum_{j=1}^\infty \Delta^2 f(j) \cdot (\load - j + 1)^+
		~.
	\end{align*}
	
	The lemma then follows by substituting $c = j-1$, and noting that $\Delta^2 f(1) = 0$.
%
%
%

\subsection{Proof of \Cref{cor:convex-amortized}}
\label{app:convex-amortized}

By \Cref{lem:convex-decomposition} and $(L - c)^+ = L - \min \{L, c\}$, we have:
\begin{align*}
	\E \sum_{v \in V} f(\load_v)
	&
	~=~ 
	\E \sum_{v \in V} \sum_{c=1}^\infty \Big( f(c-1) + f(c+1) - 2 f(c) \Big) \cdot \Big( \load_v - \min \big\{\load_v , c\big\} \Big)
	\\
	&
	~=~
	\sum_{c=1}^\infty \Big( f(c-1) + f(c+1) - 2 f(c) \Big) \E \sum_{v \in V} \load_v \\
	& \qquad\qquad
	- \sum_{c=1}^\infty \Big( f(c-1) + f(c+1) - 2 f(c) \Big) \E \sum_{v \in V} \min \big\{\load_v , c\big\} 
	~.
\end{align*}

On one hand, $\sum_{v \in V} \load_v$ is the number of arrived balls.
Hence, its expectation equals the sum of the arrival rate of the balls, which further equals the sum of the degrees of the bins, i.e.:
\[
	\E \sum_{v \in V} \load_v \,=\, \sum_{e \in E} \lambda_e \,=\, \sum_{v \in V} d_v
	~.
\]

On the other hand, by \Cref{cor:c-matching}, we have:
\[
	\E \sum_{v \in V} \min \big\{\load_v , c\big\} \,\ge\, \E_{j \sim \mathrm{Poisson}(2d)} \frac{\min \{ j, 2c \}}{2d} \cdot\sum_{v \in V} d_v~.
\]

Putting together and grouping terms by $d_v$, we have:
\begin{align*}
	\E \sum_{v \in V} f(\load_v)
	&
	~\le~ \sum_{v \in V} d_v \sum_{c=1}^\infty \Big(  f(c-1) + f(c+1) - 2f(c) \Big) \bigg( 1 - \frac{\E_{j \sim \mathrm{Poisson}(2d)} \min \{ j, 2c \}}{2d} \bigg) \\
	&
	~=~ \sum_{v \in V} d_v \sum_{c=1}^\infty \Big(  f(c-1) + f(c+1) - 2f(c) \Big) \cdot \frac{\E_{j \sim \mathrm{Poisson}(2d)} \big[ j - \min \{ j, 2c \} \big]}{2d} \\
	&
	~=~	\sum_{v \in V} d_v \sum_{c=1}^\infty \Big(  f(c-1) + f(c+1) - 2f(c) \Big) \cdot \frac{\E_{j \sim \mathrm{Poisson}(2d)} \big[ (\lfloor \nicefrac{j}{2} \rfloor - c)^+ + (\lceil \nicefrac{j}{2} \rceil - c)^+ \big]}{2d}
	\\
	&
	~=~ \sum_{v \in V} d_v \cdot \E_{j \sim \mathrm{Poisson}(2d)} \frac{f( \lfloor \nicefrac{j}{2} \rfloor ) + f( \lceil \nicefrac{j}{2} \rceil )}{2d}
	~.
\end{align*}
Here the last step follows by applying the decomposition in \Cref{lem:convex-decomposition} in the opposite direction.

\bigskip
\noindent
\textbf{Remark.}
{Similar to the discussion of maximizing a concave function (\Cref{app:concave-amortized}), we
define $\Delta f(j \,\vert\, S_e)$ to be the minimum increase of function values if we allocate $j$ balls to a ball $e$'s incident bins given the load pair $S_e$.
Formally:
\[
	\Delta f(j \,\vert\, S_e) 
	\,\defeq\,
	\min_{x \,:\, \sum_{v \sim e} x_v = j} ~ \sum_{v \sim e} \big( f(\load_v + x_v) - f(\load_v) \big)
	~.
\] 
This is the total increase of the function values of bins incident to $e$ when we allocate $j$ balls optimally to them subject to their existing loads.
We can construct such an optimal allocation, e.g., by allocating j balls one by one, each to the less loaded bin. Define the potential of a ball $e$ with load pair $S_e$ as:
\[
	\Phi^{S_e}(t) \,=\, \frac{1}{2d} \E_{j \sim \mathrm{Poisson}(2d(1-t))} \Delta f(j \,\vert\,S_e)
	~.
\]
The corollary then follows by verifying the conditions of \Cref{thm:potential-based-general} with the above potentials.
}

\subsection{Proof of \Cref{cor:load-balancing} and \Cref{cor:completion-time-min}}
\label{app:convex-function-examples}
\begin{corollary}
     \label{cor:general-convex-amortized}
    For a convex function $f : \Z_{\ge 0} \to \R_{\ge 0}$ not normalized,
    decompose it into $f(L) = f_*(L) + r(L)$,
    where
    $f_*(\load) = f(\load)  - \big(f(1) - f(0)\big) \cdot \load - f(0)$ and
    $r(\load) = \big(f(1) - f(0)\big) \cdot \load + f(0)~.$
    Then Greedy guarantees an amortized bound:
    \[
    	g(d_v) ~=~ \E_{j \sim \mathrm{Poisson}(2d)} \frac{f_*( \lfloor \nicefrac{j}{2} \rfloor ) + f_* ( \lceil \nicefrac{j}{2} \rceil )}{2d}
        \cdot d_v + r(d_v)
        ~.
    \]
    That is, the expected objective is at most:
    %
    \begin{equation*}
        \E \sum_{v \in V}f(L_v) ~\le~
        \E_{j \sim \mathrm{Poisson}(2d)} \frac{f_*( \lfloor \nicefrac{j}{2} \rfloor ) + f_* ( \lceil \nicefrac{j}{2} \rceil )}{2d}
        \cdot \sum_{v \in V} d_v + \sum_{v \in V} r(d_v)~.
    \end{equation*}
\end{corollary}
\begin{proof}
On the one hand, it is easy to verify that 
$f_*$
is normalized ($f_*(0) = f_*(1) = 0$) and convex.

On the other hand, the total contribution of $r$ to the original function $f$ does not depend on the decision process of the algorithm since:
\begin{equation*}
    \sum_{v \in V} r(\load_v) = |V| \cdot f(0) + \big(f(1) - f(0) \big)\cdot \sum_{v \in V} \load_v~,
\end{equation*}
where $\sum_{v\in V} \load_v$ equals the number of balls arrived. Therefore, the amortized optimality of Greedy algorithm in \Cref{cor:convex-amortized} continues to hold for $f$.

Applying \Cref{cor:convex-amortized} to the normalized function $f_*$ gives:
\begin{equation*}
        \E \sum_{v \in V}f(L_v) ~\le~
        \E_{j \sim \mathrm{Poisson}(2d)} \frac{f_*( \lfloor \nicefrac{j}{2} \rfloor ) + f_*( \lceil \nicefrac{j}{2} \rceil )}{2d}
        \cdot \sum_{v \in V} d_v + \sum_{v \in V} r(d_v)~.  \qedhere
\end{equation*}
\end{proof}

\noindent
\textit{Proof of \Cref{cor:load-balancing}}~
Using the decomposition in \Cref{cor:general-convex-amortized} gives
\begin{equation*}
    f_*(\load) = e^{\load} - (e-1)L - 1~,~r(\load) = (e-1)\load +1~.   
\end{equation*}

Applying \Cref{cor:general-convex-amortized}, we have:
\begin{align*}
    \E \sum_{v \in V}f(L_v) ~&\le~
    \E_{j \sim \mathrm{Poisson}(2d)} \bigg[ \frac{1}{2d} \Big(e^{\lfloor \nicefrac{j}{2} \rfloor} - (e-1)\lfloor \nicefrac{j}{2} \rfloor - 1 
    +
    e^{\lceil \nicefrac{j}{2} \rceil} - (e-1)\lceil \nicefrac{j}{2} \rceil -1 \Big) \bigg] \cdot \sum_{v \in V}d_v \\
    ~&~~~~~+~ \sum_{v \in V} \Big( (e-1)d_v + 1 \Big)~ \\
    ~&=~
    \E_{j \sim \mathrm{Poisson}(2d)}
    \frac{e^{\lfloor \nicefrac{j}{2} \rfloor} + e^{\lceil \nicefrac{j}{2} \rceil} - (e-1)(j-2d) - 2}{2d} \cdot \sum_{v\in V}d_v + |V|~. 
\end{align*}

\noindent
\textit{Proof of \Cref{cor:completion-time-min}}~
    Using the decomposition in \Cref{cor:general-convex-amortized} gives
    \begin{equation*}
        f_*(\load) = \frac{{\load}^2 - \load}{2}~,~r(\load) = \load~.
    \end{equation*}
    
    Applying \Cref{cor:general-convex-amortized}, we have:
    \begin{align*}
        \E \sum_{v \in V}f(L_v)
        ~&\le~
        \E_{j \sim \mathrm{Poisson}(2d)} \bigg[ \frac{{\lfloor \nicefrac{j}{2} \rfloor}^2 - {\lfloor \nicefrac{j}{2} \rfloor} + {\lceil \nicefrac{j}{2} \rceil}^2 - {\lceil \nicefrac{j}{2} \rceil}}{2 \cdot 2d} \Bigg] \cdot \sum_{v \in V}d_v + \sum_{v \in V}d_v \\
        ~&=~
        \E_{j \sim \mathrm{Poisson}(2d)} \frac{{\lfloor \nicefrac{j}{2} \rfloor}^2 + {\lceil \nicefrac{j}{2} \rceil}^2 - j + 4d}{4d} \cdot \sum_{v \in V}d_v~. \qedhere
    \end{align*}

\section{Deferred Proofs from \Cref{sec:mark-swap}}
\label{app:mark-swap}

\subsection{Proof of \Cref{lem:swap-concave}}
\label{app:swap-concave}

By the definition of second-order stochastic dominance, we need to show for any concave function $f$ that:
\[
 	f \big( \Swap(\alpha) + \Swap(\beta) \big) \,\ge\, f \big( \Swap(\alpha - \varepsilon) + \Swap(\beta + \varepsilon) \big)
 	~,
\]
where recall that $f(D) = \E_{X \sim D} f(X)$.

By the definition of distributions $\Swap(\lambda)$:
\begin{align*}
	f \big( \Swap(\alpha) + \Swap(\beta) \big) 
	&
	= \E_{X, Y \sim \Exponential(1)} f \Big( \Poisson \big( \alpha (1-X)^+ \big) + \Poisson \big( \beta (1-Y)^+ \big) \Big) \\
	&
	= \E_{X, Y \sim \Exponential(1)} f \Big( \Poisson \big( \alpha (1-X)^+ +  \beta (1-Y)^+ \big) \Big)
	~.
\end{align*}

For ease of exposition, we define an auxiliary function:
\[
    g(x) \defeq f\big(\Poisson(x)\big)
    ~.
\]

Then, we have:
\begin{align*}
	f \big( \Swap(\alpha) + \Swap(\beta) \big)
	&
	= \E_{X, Y \sim \Exponential(1)} g \big( \alpha (1-X)^+ +  \beta (1-Y)^+ \big) \\
	&
	= \frac{1}{2} \E_{X, Y \sim \Exponential(1)} \Big[ g \big( \alpha (1-X)^+ +  \beta (1-Y)^+ \big) + g \big( \alpha (1-Y)^+ +  \beta (1-X)^+ \big) \Big]
	~,
\end{align*}
where the second equality holds because $X$ and $Y$ are independently and identically distributed.
Rewriting $f \big( \Swap(\alpha-\varepsilon) + \Swap(\beta+\varepsilon) \big)$ also in this way, the desired inequality becomes:
\begin{align*}
	&
	\E_{X, Y \sim \Exponential(1)} \Big[ g \big( \alpha (1-X)^+ +  \beta (1-Y)^+ \big) + g \big( \alpha (1-Y)^+ +  \beta (1-X)^+ \big) \Big] \\
	& \quad
	\ge
	\E_{X, Y \sim \Exponential(1)} \Big[ g \big( (\alpha-\varepsilon) (1-X)^+ + (\beta+\varepsilon) (1-Y)^+ \big) + g \big( (\alpha-\varepsilon) (1-Y)^+ + (\beta+\varepsilon) (1-X)^+ \big) \Big]
	~.
\end{align*}

We will prove the inequality for every realization of $X$ and $Y$.
Recall that $\alpha \le \beta$.
Further, suppose without loss of generality that $(1-X)^+ \ge (1-Y)^+$.
On one hand, we have:
\begin{align*}
	&
	(\alpha-\varepsilon) (1-X)^+ + (\beta+\varepsilon) (1-Y)^+  \le \alpha (1-X)^+ +  \beta (1-Y)^+ \big) \\
	& \quad
	\le \alpha (1-Y)^+ +  \beta (1-X)^+ \le (\alpha-\varepsilon) (1-Y)^+ + (\beta+\varepsilon) (1-X)^+
	~.
\end{align*}

On the other hand:
\begin{align*}
	&
	\big( \alpha (1-X)^+ +  \beta (1-Y)^+ \big) + \big( \alpha (1-Y)^+ +  \beta (1-X)^+ \big) \\
	& \quad
	=
	\big( (\alpha-\varepsilon) (1-X)^+ + (\beta+\varepsilon) (1-Y)^+ \big) + \big( (\alpha-\varepsilon) (1-Y)^+ + (\beta+\varepsilon) (1-X)^+ \big)
	~.
\end{align*}

Therefore, the inequality reduces to the concavity of function $g$.
Consider any $x > y \ge 0$, and random variables $A \sim \Poisson(y)$ and $B, C \sim \Poisson(\nicefrac{(x-y)}{2})$.
We have:
\begin{align*}
	g(x) + g(y) 
	&
	= \E \big[ f(A) + f(A + B + C) \big] 
	\tag{$\Poisson(y) + 2 \Poisson(\nicefrac{(x-y)}{2}) = \Poisson(x)$} \\[1ex]
	&
	\le \E \big[ f(A + B) + f(A + C) \big] 
	\tag{concavity of $f$} \\
	&
	= 2 \cdot g\Big(\frac{x+y}{2} \Big)
	~.
	\tag{$\Poisson(y) + \Poisson(\nicefrac{(x-y)}{2}) = \Poisson(\nicefrac{(x+y)}{2})$}
\end{align*}

Hence, function $g$ is indeed concave.

\subsection{Proof of \Cref{lem:markswap-dominate-poisson}}
\label{app:markswap-dominate-poisson}
	Consider random variables $T \sim \Exponential(1)$ and $S \sim \Poisson\big( \lambda (1-T)^+ \big)$ so that $S \sim \Swap(\lambda)$.
	Further, sample $F' \sim \Bernoulli(\lambda)$ if $T < 1$, and let $F' = 0$ otherwise.
	
	On one hand, $T$ can be seen as the arrival time of a Poisson process with rate $1$ and $F'$ is the indicator for a sub-sample with probability $\lambda$;
	the combined process corresponds to a Poisson process with rate $\lambda$.
	Since $S$ is the number of Poisson arrivals with rate $\lambda$ after time $T$, we conclude that $S + F' \sim \Poisson(\lambda)$.
	
	On the other hand, we also have $F' \sim \Bernoulli\big( (1-\frac{1}{e}) \lambda \big)$, except that it is positively correlated with $S$.
	More precisely, for any $c \ge 1$ we have:
	\[
		\Pr \big[ S \ge c \,\vert\, F' = 0 \big] 
		=
		\frac{(1-\lambda) \int_0^1 e^{-t} \cdot \Pr_{\load \sim \Poisson(\lambda(1-t))} \big[ \load \ge c \big] \,\dif{t}}{(1-\lambda) \int_0^1 e^{-t}  \,\dif{t} + \frac{1}{e}}
	\]
	by Bayes' Rule, and similarly:
	\[
		\Pr \big[ S \ge c \,\vert\, F' = 1 \big] 
		=
		\frac{\lambda \int_0^1 e^{-t} \cdot \Pr_{\load \sim \Poisson(\lambda(1-t))} \big[ \load \ge c \big] \,\dif{t}}{\lambda \int_0^1 e^{-t}  \,\dif{t}}
		~.
	\]
	The letter is larger for every $c$.

	Compared to $F \sim \Bernoulli\big( (1-\frac{1}{e}) \lambda \big)$ drawn independent to the realization of $S$, we conclude that  $S + F \succeq S + F'$.
%

\section{Deferred Proofs from \Cref{sec:completion-time}}

\subsection{Proof of \Cref{lem:completion-time-lp}}
\label{app:completion-time-lp}

Consider any instance and its offline optimal solutions for different realizations of jobs, it suffices to construct a feasible LP solution with the same objective value.
%
Let $X_{ij}$ be the number of jobs of type $j$ allocated to machine $i$. 
Let $Y_{ijk}(s)$ be the indicator that, among jobs of sizes at least $s$ allocated to machine $i$, the $k$-th largest job is of type $j$.
Finally, let $Z_{ik}(s)$ be the indicator that machine $i$ is assigned with at least $k$ jobs of sizes at least $s$.
We will let $x_{ij} = \E X_{ij}$, $y_{ijk}(s) = \E Y_{ijk}(s)$, and $z_{ik}(s) = \E Z_{ik}(s)$.

We first verify that the LP objective of these variables equals the expected total completion time of the offline solution.
Under any realization of jobs and the corresponding optimal allocation of jobs to machines, we claim that the total completion time of jobs on machine $i$ is equal to:
\[
    \int_0^\infty \sum_{k=1}^\infty k \cdot Z_{ik}(s)\,\dif{s}
    ~.
\]

Suppose $n$ jobs were allocated to machine $i$ with sizes $s_1 \le s_2 \le \dots \le s_n$.
Then, the total completion time is:
\begin{multline}
    \label{eqn:lp-relaxation-objective}
    s_1 + (s_1 + s_2) + \dots + (s_1 + s_2 + \dots + s_n) \\
    = s_1 \cdot \frac{n^2+n}{2} + (s_2 - s_1) \cdot \frac{(n-1)^2 + (n-1)}{2} + \dots + (s_n - s_{n-1}) \cdot \frac{1^2 + 1}{2}
\end{multline}

Artificially let $s_0 = 0$ and $s_{n+1} = \infty$ for notational convenience.
For any $0 \le i \le n$ and any $s_i \le s < s_{i+1}$, there are $n-i$ jobs of sizes at least $s$.
Hence, we have $Z_{ik}(s) = 1$ for $1 \le k \le n-i$ and $Z_{ik}(s) = 0$ otherwise.
As a result:
\[
    \sum_{k=1}^\infty k \cdot Z_{ik}(s) ~=~ \sum_{k=1}^{n-i} k ~=~ \frac{(n-i)^2 + (n-i)}{2}
    ~.
\]
Integrating over $s \ge 0$, this equals \Cref{eqn:lp-relaxation-objective}.

Next, we verify the constraints.
The first set of constraints holds since $\sum_{i \in M} X_{ij}$ equals the number of realized jobs of type $j$, whose expectation is $\lambda_j$.

The second set of constraints holds for every realization, i.e., $\sum_{k=1}^\infty Y_{ijk}(s) = X_{ij}$, since every job $j$ allocated to machine $i$ is a job of size $s_{ij} \ge s$.
Therefore, there is a one-to-one mapping between such jobs and variables $Y_{ijk}(s) = 1$,  by the definition of $Y_{ijk}(s)$.

The third set of constraints follows by two upper bounds of $\sum_{j \in S} \sum_{k \le \ell} Y_{ijk}(s)$.
On one hand, this is at most the number of realized jobs with types in $S$, which follows distribution $\Poisson(\lambda_S)$.
On the other hand, it is at most $\ell$, because for any $1 \le k \le \ell$, there is at most one type of jobs $j$ for which $Y_{ijk}(s) = 1$.

Finally, the fourth set of constraints holds for every realization, i.e., $Z_{ik}(s) = \sum_{j \in J} Y_{ijk}(s)$, by the definitions of these variables.

\subsection{Proof of \Cref{lem:lp-poly-solvable}}
\label{app:lp-poly-solvable}

	The main challenge is that the LP has a continuum of variables and constraints.
	We first reduce the number of variables from a continuum to a polynomial of the problem size.
	Note that there are at most $|M| |J|$ distinct job sizes, denoted as $0 \le s_1 \le s_2 \le \dots \le s_{|M||J|}$.
	Artificially define $s_0 = 0$ and $s_{|M||J|+1} = \infty$ for notational convenicence.
	Then, we may without loss of generality consider solutions satisfying that for any $0 \le \ell \le |M||J|$, variables $z_{ik}(s)$ and $y_{ijk}(s)$ have fixed values for $s_\ell \le s < s_{\ell+1}$.
	
	With this treatment, the number of variables is polynomial in $|M|$ and $|J|$.
	Further, it suffices to consider the second and third sets of constraints for a polynomial number of size-levels $s$.
	That being said, there are still exponentially many constraints in the third set, even for a fixed size-level $s = s_i$, $0 \le i \le |M||J|$.
	Fortunately, for any fixed $i$, $\ell$, and $s$, the right-hand-side of the constraint is submodular, and thus, the corresponding constraints form a poly-matroid.
	Hence, there is a polynomial time separation oracle for the third set of constraints, allowing us to solve the LP in polynomial time, e.g., using the ellipsoid method.

\section{Deferred Proofs from \Cref{sec:discussion}}
\subsection{Proof of \Cref{thm:potential-based-general-extention}}
\label{app:general-framework-extension}
	We prove the theorem by verifying the conditions of \Cref{lem:potential-function-conditions}.
	The starting and terminal conditions of \Cref{lem:potential-function-conditions} follow from the starting and terminal conditions of the theorem.
	
	The rest of the proof verifies the monotonicity condition of \Cref{lem:potential-function-conditions}, using the last two conditions of the theorem.
	We only prove the maximization version, as the minimization version is almost verbatim.
	
	It suffices to show that $\frac{\dif}{\dif{t}} \E \big[\Phi(t)+A(t)\big] \ge 0$ for $t \in [0, 1]$.
	We prove a stronger property that the inequality holds \emph{conditioned on what happened before time $t$}, in particular, on the realization of the balls' load tuples, represented by the normalized arrival rates $\Lambda^{S}_v$ of balls $e \sim v$ with $S_e = S$.

	At time $t$, the algorithm allocates a ball to bin $v$ with rate:
	\[
		\sum_{S \in \loadeset_v} \Lambda^{S}_v \, |S| \,x^{S}_v
		~,
	\]
	where $\Lambda^{S}_v \,|S|$ is the (unnormalized) arrival rate of balls $e \sim v$ with load tuples $S_e = S$, and $x^{S}_v$ is the probability that the algorithm allocates such a ball to bin $v$ conditioned on its arrival.
	
	Hence, the algorithm's objective changes by:
	\begin{equation}
		\label{eqn:general-algo}
		\frac{\dif{}}{\dif{t}} \E A(t) ~=~ \sum_{v \in V} \sum_{S \in \loadeset_v} \Lambda^{S}_v \,|S|\, x^{S}_v \, \Big( f(\load_v + 1) - f(\load_v) \Big)
		~.
	\end{equation}

	Further, when the algorithm allocates a ball to bin $v$, its potential decreases by:
	\[
		\sum_{S \in \loadeset_v} \Lambda^{S}_v \, |S| \, \Big( \Phi^{S}(t) - \Phi^{S_+}(t) \Big)
		~.
	\]

	Summing over $v \in V$, the decrease of the overall potential due to the allocation at time $t$ is:
	\begin{equation}
	    \label{eqn:general-pot-allocation}
	    \sum_{v\in V}\,
	    \biggl(\, 
	    	\sum_{S \in \loadeset_v} \Lambda^{S}_v \,|S|\, x^{S}_v 
		\,\biggr)
	    \biggl(\,
	    	\sum_{S \in \loadeset_v} \Lambda^{S}_v \,|S| \, \Big( \Phi^{S}(t) - \Phi^{S_+}(t) \Big) 
    	\,\biggr)
	    ~.
	\end{equation}

	Finally, the overall potential also changes due to time elapses, independent of the arrival and allocation at time $t$, by:
	\begin{equation}
	    \label{eqn:general-pot-time}
	    \sum_{e \in E} \lambda_e \cdot\frac{\dif{}}{\dif{t}} \Phi^{S_e}(t)
		\,=\,
		\sum_{\vphantom{S} v \in V} \sum_{S \in \loadset_v} \Lambda^{S}_v \,|S| \, x^{S}_v \cdot \frac{\dif{}}{\dif{t}} \Phi^{S}(t)
	    ~.
	\end{equation}

	Putting \Cref{eqn:general-algo,eqn:general-pot-allocation,eqn:general-pot-time} together and grouping the linear terms w.r.t.\ $\Lambda^{S}_v$, we get that $\frac{\dif{}}{\dif{t}} \E \big[\Phi(t)+A(t)\big]$ is equal to:
	\[
		\sum_{v \in V} 
		\Bigg( 
		    \sum_{S \in \loadeset_v} \Lambda_v^{S} \, |S| \, x^{S}_v \,\Big( f(\load_v + 1) - f(\load_v) + \frac{\dif}{\dif{t}} \Phi^{S}(t) \Big)
			\,-\,
			\sum_{S \in \loadeset_v} \Lambda_v^{S} \, |S| \, x^{S}_v
			\sum_{S \in \loadeset_v} \Lambda_v^{S} \, |S| \, \Big( \Phi^{S}(t) - \Phi^{S_+}(t) \Big)
		\Bigg)
	    ~.
	\]

	It suffices to show non-negativity for every bin $v \in V$.
	Recall that $\vec{\Lambda}_v$ lies in polytope $\polytope_v$.
	In other words, there are $\alpha_{\vec{\Lambda}} \ge 0$ for $\vec{\Lambda} \in \support(\polytope_v)$ satisfying $\sum_{\vec{\Lambda}} \alpha_{\vec{\Lambda}} = 1$ such that:
	\[
		\vec{\Lambda}_v = \sum_{\vec{\Lambda} \in \support(\polytope_v)} \alpha_{\vec{\Lambda}} \cdot \vec{\Lambda}
		~.
	\]
	
	The inequality for bin $v$ is then:
	\begin{align*}
		&
		\sum_{\vec{\Lambda} \in \support(\polytope_v)} \alpha_{\vec{\Lambda}} \sum_{S \in \loadeset_v} \Lambda^{S} \, |S| \, x^{S}_v \Big( f(\load_v + 1) - f(\load_v) + \frac{\dif}{\dif{t}} \Phi^{S}(t) \Big) \\
		& \qquad
		\ge\, 
		\bigg( \sum_{\vec{\Lambda} \in \support(\polytope_v)} \alpha_{\vec{\Lambda}} \sum_{S \in \loadeset_v} \Lambda^{S} \, |S| \, x^{S}_v \bigg)
		\bigg( \sum_{\vec{\Lambda} \in \support(\polytope_v)} \alpha_{\vec{\Lambda}} \sum_{S \in \loadeset_v} \Lambda^{S} \, |S| \, \Big( \Phi^{S}(t) - \Phi^{S_+}(t) \Big) \bigg)
		~.
	\end{align*}
	
	By Chebyshev's sum inequality, which is applicable due to the last condition of the theorem, the right-hand-side is at most:
	\[
		\sum_{\vec{\Lambda} \in \support(\polytope_v)} \alpha_{\vec{\Lambda}} 
		~
		\sum_{S \in \loadeset_v} \Lambda^{S} \, |S| \, x^{S}_v 
		~
		\sum_{S \in \loadeset_v} \Lambda^{S} \, |S| \, \Big( \Phi^{S}(t) - \Phi^{S_+}(t) \Big)
		~.
	\]
	
	Hence, the inequality follows by the third condition of the theorem.

\subsection{Proof of \Cref{cor:k-choice}}
\label{app:k-choice}

Recall the potential functions:
\[
	\Phi^{S_e}(t) = \E_{j \sim \mathrm{Poisson}(dk(1-t))} \frac{\min \bigl\{ j \,,\, \sum_{v \sim e} (c - \load_v)^+ \bigr\}}{dk}
	~.
\]

The starting condition at $t = 0$ holds because for any bin $v$:
\[
    \Phi^{\mathbf{0}}(0) 
    \,=\,
    \E_{j \sim \mathrm{Poisson}(dk)}\frac{\min \{ j, ck \}}{dk}
    ~.
	\]
	
	The terminal condition also holds because $\Phi^{S}(1) = 0$ for any load tuple $S$ by definition.

It remains to verify the last two conditions of \Cref{thm:potential-based-general-extention}.
Consider an arbitrary bin $v \in V$ and its load $\load_v$.
The subset of consistent load tuples are:
\[
	\loadeset_v \,=\, \Big\{ (\load_v, \load_1, \load_2, \dots, \load_{k-1}) : 0 \le \load_i \le c \mbox{ for } 1 \le i \le k-1 \Big\}
	~,
\]
where load tuples with $\load_i > c$ are omitted since the argument can treat them as $\load_i = c$.

By the definition of the Proportional Greedy algorithm, the probability of allocating the ball $e$ to bin $v$ only depends on the remaining capacity $r_v = (c - \load_v)^+$ of bin $v$ and the total remaining capacity of all incident bins of ball $e$, denoted as $r_e = \sum_{u \sim e} (c - \load_u)^+$.
The same applies to the potential of ball $e$.
Hence, we will write $r_e$ instead of the entire load tuple $S$ as the superscripts in the following notations:
\[
	x_v^r = \frac{r_v}{r} 
	~,\quad
	\Phi^r(t) = \E_{j \sim \mathrm{Poisson}(dk(1-t))} \frac{\min \{ j , r \}}{dk}
	~.
\]

Let $\Lambda_v^r = \frac{1}{k} \sum_{e \sim v : r_e = r} \lambda_e$ is the normalized total arrival rate of balls incident to bin $v$ with total remaining capacity being $r_e = r$. 
The vector $\vec{\Lambda} = (\Lambda^r)_{r_v \le r \le c(k-1)+r_v}$ lies in the polytope:
\[
	\polytope_v = \bigg\{\, \vec{\Lambda} : \sum_{r=r_v}^{c(k-1)+r_v} \Lambda^r = d_v \,\mbox{, and}~ \Lambda^r \ge 0 ~\mbox{for $r_v \le r \le c(k-1)+r_v$} \,\bigg\}
	~.
\]

It is spanned by the vectors with $\Lambda^r = d_v$ for some $r$ and zeros in the other entries.
The case of $r = 0$ is trivially true since $\Phi^{\mathbf{0}}(t) = 0$ at all time $t$.

Next, suppose $r\ge 1$.
For the vector with $\Lambda^r = d_v$, the inequality reduces to:
\[
	d_v k \cdot \frac{r_v}{r} \cdot \bigg( \mathbf{1}_{r_v > 0} + \frac{\dif{}}{\dif{t}} \Phi^r (t) \bigg)
	\,\ge\,
	d_v k \cdot \frac{r_v}{r} \cdot  d_v k \, \Big( \Phi^r(t) - \Phi^{r-1}(t) \Big)
	~.
\]

If $r_v = 0$, both sides are zero.
Otherwise, we cancel the common $d_vk\cdot \frac{r_v}{r}$ term from both sides.
Since the right-hand-side is non-negative and $d_v \le d$, we only need to verify the case of $d_v = d$, where it holds with equality, i.e.:
\begin{equation}
	\label{eqn:k-choice-c-matching-monotone}
	1 + \frac{\dif{}}{\dif{t}} \Phi^r(t)
	\,=\,
	dk \, \Big( \Phi^r(t) - \Phi^{r-1}(t) \Big)
	~.
\end{equation}

By the definition of $\Phi^r$ and as a corollary of \Cref{lem:cap-expectation-change-rate} with $k = r$ and $\lambda = dk(1-t)$, the left-hand-side of \Cref{eqn:k-choice-c-matching-monotone} equals:
\[
	1 - \Pr_{j \sim \mathrm{Poisson}(dk(1-t))} \big[\, j < r \,\big] = \Pr_{j \sim \mathrm{Poisson}(dk(1-t))} \big[\, j \ge r \,\big]
	~,
\]
which equals the right-hand-side of \Cref{eqn:k-choice-c-matching-monotone} as a corollary of \Cref{lem:cap-expectation-change-cap} for the same $k = r$ and distribution $D = \Poisson\big(dk(1-t)\big)$.

\subsection{Proof of \Cref{cor:Mixed-choice}}
\label{app:mixed-choice}

Consider potential functions:
\begin{align*}
    \Phi^{\mathbf{0}} (t) &=
    1 + \frac{\ln{2}}{1-\ln{2}}e^{-2(1-t)} - \frac{1}{1-\ln{2}}(2e)^{-(1-t)}~, \\
    \Phi^{(0,1)}(t) &=
    \frac{1}{2} - \frac{1}{2}e^{-2(1-t)}~, \\[1ex]
    \Phi^{(1,1)}(t) &= \Phi^{(1)}(t) = 0~.
\end{align*}

The starting condition at $t = 0$ holds because for any bin $v$:
\[
    \Phi^{\textbf{0}}(0) = 1 + \frac{\ln{2}}{1-\ln{2}} \frac{1}{e^2} - \frac{1}{1-\ln{2}} \frac{1}{2e}
    ~.
\]

The terminal condition also holds because $\Phi^{S}(1)=0$ for any load tuple $S$.

It remains to verify the last two conditions of \Cref{thm:potential-based-general-extention}.
Recall that $f(\load) = \min \{ \load , 1 \}$.
Consider an arbitrary bin $v \in V$ and its load $\load_v \in \{ 0, 1 \}$.
	The set of consistent load tuples is:
	\[
		\loadeset_v \,=\, \Big\{ (\load_v, i) : i = 0,1 \Big\}
  \bigcup \Big\{ (\load_v) \Big\}~.
	\]

Recall that $\Lambda_v^{S} = \frac{1}{|S|} \sum_{e \sim v : S_e = S} \lambda_e$.
Consider any empty bin $v$ and its vector of arrival rates $\vec{\Lambda}_v = (\Lambda_v^{(0)} , \Lambda_v^{(0,0)}, \Lambda_v^{(0,1)})$.
Vector $\vec{\Lambda}_v$ lies in the polytope:
\[
    \polytope_v = \bigg\{ \vec{\Lambda} : 
{||\vec{\Lambda}||}_1 = d_v~,~ 
\Lambda^{(0)} \le 1-\ln2 \mbox{~and}~ \Lambda^{(0)},\Lambda^{(0,0)},\Lambda^{(0,1)} \ge 0 \bigg\}
    ~.
\]


Recall the third condition of \Cref{thm:potential-based-general-extention}.
For any $\vec{\Lambda} \in \support\big( \polytope_v \big)$, we need to show that:
\[
    \sum_{S \in \loadeset_v} \Lambda^{S} |S| \, x^{S}_v \Big( 1 + \frac{\dif}{\dif{t}} \Phi^{S}(t) \Big) 
    \,\ge\, 
    \sum_{S\in \loadeset_v} \Lambda^{S} |S| \, x^{S}_v \sum_{S \in \loadeset_v} \Lambda^{S} |S| \, \Big( \Phi^{S}(t) - \Phi^{S_+}(t) \Big)
    ~.
\]

We may consider without loss of generality $d_v > 0$.
If $d_v \le 1 - \ln 2$, the spanning vectors are $(d_v, 0, 0)$, $(0, d_v, 0)$, and $(0, 0, d_v)$.
It is easy to verify that (with $d_v = 1 - \ln 2$ as the worst case):
\begin{enumerate}[label=(\roman*)]
\item For vector $(0,d_v,0)$, the inequality is:
\[
    1+\frac{\dif}{\dif{t}} \Phi^{\textbf{0}}(t) 
    \,\ge\,
    2 d_v \Big(\Phi^{\textbf{0}}(t) - \Phi^{(0,1)}(t) \Big)
    ~.
\]
\item For vector $(d_v,0,0)$, the inequality is:
\[
    1+\frac{\dif}{\dif{t}} \Phi^{\textbf{0}}(t) 
        \,\ge\,
        d_v \, \Phi^{\textbf{0}}(t)
    ~.
\]
\item For vector $(0,0,d_v)$, the inequality is:
\[
    1+\frac{\dif}{\dif{t}} \Phi^{(0,1)}(t)
    \ge
    2 d_v \, \Phi^{(0,1)}(t)
    ~.
\]
\end{enumerate}

The fourth condition of \Cref{thm:potential-based-general-extention} follows by the non-decreasingness of the second and third rows below, where we use the fact that $\Phi^{\textbf{0}} \le 2 \Phi^{(0,1)}(t)$.
\begin{table}[h]
\centering
\begin{tabular}{lx{3.6cm}x{3cm}x{3cm}}
\toprule
& (i) & (ii) & (iii) \\
\midrule
$\sum_{S \in \loadeset_v} \Lambda^{S} |S| \, x^{S}_v$ & $d_v$ & $d_v$ & $2d_v$ \\
$\sum_{S \in \loadeset_v} \Lambda^{S} |S| \, \big( \Phi^{S}(t) - \Phi^{S_+}(t) \big)$ & $2d_v \big( \Phi^{\textbf{0}}(t) - \Phi^{(0,1)}(t) \big)$ & $d_v \Phi^{\textbf{0}}(t)$ & $2d_v \Phi^{(0,1)}(t)$\\
\bottomrule
\end{tabular}
\end{table}

If $d_v > 1 - \ln 2$, we list the spanning vectors $(1-\ln2,d_v -1 + \ln 2, 0)$ and the corresponding inequalities below.
It is easy to verify the inequalities in the worst case when $d_v = 1$.
\begin{enumerate}[label=(\roman*)]
    \setcounter{enumi}{3}
    \item  \label{case:1}
    For vector $(0, d_v, 0)$, the inequality is:
    \[
        1+\frac{\dif}{\dif{t}} \Phi^{\textbf{0}}(t) 
        \,\ge\,
        2 d_v \Big(\Phi^{\textbf{0}}(t) - \Phi^{(0,1)}(t) \Big)
        ~.
    \]
    
    \item \label{case:2}
    For vector $( 1-\ln2 , d_v - 1 + \ln2 , 0)$, the inequality is:
    \[
        1+\frac{\dif}{\dif{t}} \Phi^{\textbf{0}}(t)
        \ge
        \big(2 d_v -1+\ln2\big) \Phi^{\textbf{0}}(t)-2 \big( d_v - 1 + \ln{2} \big)  \Phi^{(0,1)}(t)
        ~.
    \]
    
    \item \label{case:3}
    For vector $(1-\ln2 , 0 , d_v - 1 + \ln2)$, the inequality is:
    \begin{multline*}
        (1-\ln2)\Big(1+\frac{\dif}{\dif{t}} \Phi^{\textbf{0}}(t)\Big) +
        2 \big( d_v - 1 + \ln2 \big) \Big(1+\frac{\dif}{\dif{t}} \Phi^{(0,1)}(t)\Big) \\
        \ge (2 d_v - 1 +\ln2)
        \Big((1-\ln2)\Phi^{\textbf{0}}(t) + 2 (d_v - 1 + \ln2) \Phi^{(0,1)}(t) \Big) ~.
    \end{multline*}
    This is quadratic in $d_v$ with a positive quadratic coefficient on the right.
    Hence, it suffices to verify $d_v = 1 - \ln 2$ and $d_v = 1$.
    Since we have already check $d_v = 1 - \ln 2$ in the previous case, we only need to verify $d_v = 1$ as for the other three cases.
    \item \label{case:4}
    For the vector $(0, 0, d_v)$, the inequality is:
    \[
        1+\frac{\dif}{\dif{t}} \Phi^{(0,1)}(t)
        \ge
       2 d_v \, \Phi^{(0,1)}(t)
        ~.
    \]
\end{enumerate}

The fourth condition of \Cref{thm:potential-based-general-extention} follows by the non-decreasingness of the second and third columns below, where we once again use the fact that $\Phi^{\textbf{0}}(t) \le 2 \Phi^{(0,1)}(t)$.
\begin{table}[h]
\renewcommand{\arraystretch}{1.25}
\centering
\begin{tabular}{cx{5cm}x{8cm}}
\toprule
& $\sum_{S \in \loadeset_v} \Lambda^{S} |S| \, x^{S}_v$ & $\sum_{S\in \loadeset_v} \Lambda^{S} |S| \, \big( \Phi^{S}(t) - \Phi^{S_+}(t) \big)$ \\
\midrule
(iv) & $d_v$ & $2d_v \big( \Phi^{\textbf{0}}(t) - \Phi^{(0,1)}(t) \big)$ \\
(v) & $d_v$ & $\big(2 d_v -1+\ln2) \Phi^{\textbf{0}}(t)-2 \big( d_v - 1 + \ln{2} \big)  \Phi^{(0,1)}(t)$ \\
(vi) & $2d_v-1+\ln2$ & $(1-\ln2)\Phi^{\textbf{0}}(t) + 2 (d_v - 1 + \ln2) \Phi^{(0,1)}(t)$ \\
(vii) & $2d_v$ & $2d_v \Phi^{(0,1)}(t)$ \\
\bottomrule
\end{tabular}
\end{table}

Therefore, by \Cref{thm:potential-based-general-extention} we get that the expected number of non-empty bins is at least:
\begin{equation*}
    \Phi^{\textbf{0}}(0) \sum_{v \in V} d_v =
    \Big( 1 + \frac{\ln{2}}{1-\ln{2}} \frac{1}{e^2} - \frac{1}{1-\ln{2}} \frac{1}{2e} \Big) \sum_{v \in V} d_v~.
\end{equation*}
\end{document}